\documentclass[12pt]{article}
\usepackage{amsmath}
\usepackage{graphicx}
\usepackage{enumerate}
\usepackage{natbib}
\usepackage{url} 

\usepackage{amsthm,amsfonts,amssymb,dsfont}
\usepackage{bigints}
\usepackage{color}
\usepackage{tikz} 
\usepackage{subfiles}
\usepackage{algorithm}
\usepackage{algpseudocode}
\usepackage{comment}
\usepackage{float}
\usepackage{longtable}
\usepackage{rotating}
\usepackage{booktabs}

\def\bF{\mathbb{F}}
\def\bbR{\mathbb{R}}
\def\SSp{\mathbb{S}_{+}}
\def\bX{\mathbf{X}}

\def\cG{\mathcal{G}}
\def\cF{\mathcal{F}}

\def\cC{\mathcal{C}}

\def\stil{\widetilde{s}}
\newcommand{\indep}{\perp \!\!\! \perp}
\def\ind{\mathbf{1}}

\newcommand\wvar[1]{w^{[#1]}}
\newcommand\Wvar[1]{W^{[#1]}}
\newcommand\Om[1]{\Omega^{[#1]}}
\def\BW{\mathbf{W}}

\newcommand\bvar[1]{b_{\lambda}^{[#1]}}
\newcommand\Bvar[1]{B_{\lambda}^{[#1]}}
\newcommand\btil[1]{\widetilde{b}_{\lambda}^{[#1]}}
\newcommand\BBvar{\mathbf{B}_{\lambda}}
\newcommand\bbobs{\mathbf{b}_{\lambda}}

\newcommand\zvar[1]{z_{\lambda}^{[#1]}}
\newcommand\Zvar[1]{Z_{\lambda}^{[#1]}}
\newcommand\ztil[1]{\widetilde{z}_{\lambda}^{[#1]}}
\newcommand\BZvar{\mathbf{Z}_{\lambda}}
\newcommand\bzobs{\mathbf{z}_{\lambda}}

\newcommand\svar[1]{s_{\lambda}^{[#1]}}
\newcommand\Svar[1]{S_{\lambda}^{[#1]}}
\newcommand\BSvar{\mathbf{S}_{\lambda}}
\newcommand\bsobs{\mathbf{s}_{\lambda}}

\newcommand\cH[1]{\mathcal{H}^{[#1]}}
\newcommand\cK[1]{\mathcal{K}^{[#1]}}

\newcommand\E[1]{E_{#1}}
\newcommand\bE[1]{\bar{E}_{#1}}
\newcommand\BE{\mathbf{E}}
\newcommand\wBE{\widetilde{\mathbf{E}}}

\newcommand\q[1]{q_{#1}}
\newcommand\bq[1]{\bar{q}_{#1}}

\newcommand\T[1]{T^{[#1]}}
\newcommand\U[1]{U^{[#1]}}
\newcommand\V[1]{V^{[#1]}}

\newcommand\Cv[1]{\Pi_s^{[#1]}}
\newcommand\iCv[1]{\Delta_s^{[#1]}}

\newcommand\bobs[1]{b^{[#1]}}

\def\wishartl{f_{\Theta}(s)}

\newcommand{\Ijkfunc}[1]{I_{j_0k_0}(#1)}

\newcommand{\Tjkfunc}[2]{T_{j_0k_0}^{[#1]}(#2)}

\newcommand{\Ujkfunc}[2]{U_{j_0k_0}^{[#1]}(#2)}
\newcommand{\univKKT}[1]{\pi^{[#1]}_{j_0k_0}}

\newcommand{\laplaceint}[1]{L_{j_0k_0}^{[#1]}}

\def\H{\mathcal{H}}


\newtheorem{theorem}{Theorem}[section]
\newtheorem{corollary}{Corollary}[theorem]
\newtheorem{lemma}[theorem]{Lemma}

\newtheorem{proposition}[theorem]{Proposition}

\newcommand{\blind}{1}

\begin{document}


\def\spacingset#1{\renewcommand{\baselinestretch}%
{#1}\small\normalsize} \spacingset{1}


\if1\blind
{
  \title{\bf {Selective Inference for Sparse Graphs via Neighborhood Selection}}
  \author{Yiling Huang\\
    Department of Statistics, University of Michigan\\
    and \\
    Snigdha Panigrahi\thanks{
    The author acknowledges support by NSF grants 1951980 and 2113342.}\\
    Department of Statistics, University of Michigan\\
    and\\
    Walter Dempsey\\
    Department of Biostatistics and Institute for Social Research,\\ University of Michigan}
  \maketitle
} \fi

\if0\blind
{
  \bigskip
  \bigskip
  \bigskip
  \begin{center}
    {\LARGE\bf Title}
\end{center}
  \medskip
} \fi

\bigskip
\begin{abstract}
Neighborhood selection is a widely used method used for estimating the support set of sparse precision matrices, which helps determine the conditional dependence structure in undirected graphical models.
However, reporting only point estimates for the estimated graph can result in poor replicability without accompanying uncertainty estimates.
In fields such as psychology, where the lack of replicability is a major concern, there is a growing need for methods that can address this issue.

In this paper, we focus on the Gaussian graphical model. 
We introduce a selective inference method to attach uncertainty estimates to the selected (nonzero) entries of the precision matrix and decide which of the estimated edges must be included in the graph. 
Our method provides an exact adjustment for the selection of edges, which when multiplied with the Wishart density of the random matrix, results in valid selective inferences. 
Through the use of externally added randomization variables, our adjustment is easy to compute, requiring us to calculate the probability of a selection event, that is equivalent to a few sign constraints and that decouples across the nodewise regressions. 
Through simulations and an application to a mobile health trial designed to study mental health, we demonstrate that our selective inference method results in higher power and improved estimation accuracy.
\end{abstract}

\noindent%
\vfill

\newpage
\spacingset{1.9} 

\section{Introduction}
The network approach to psychopathology posits mental disorders can be conceptualized as systems of co-occurring symptoms~\citep{Cramer2012, Cramer2010}. 
From this perspective, symptoms are not indicators of a latent ``common cause" but a complex network.  
Aiming to identify the inter-symptom relations, there has been a growing focus on methods for estimating conditional dependence relationships in undirected graphs. 
Such methods have provided psychologists with a useful tool for learning complex relationships between different variables.  
In psychology, these undirected graphs are often based on Gaussian graphical models (GGMs)~\citep{Epskamp2018}.  Consider, for example, the \emph{Providing Mental health Precision Treatment (PROMPT) Precision Health Study} which is a 12-month mobile health (mHealth) intervention trial focused on augmenting standard care to improve health outcomes by using mobile health technologies to extend the reach outside the clinic. A key scientific aim of PROMPT is to better understand the relationships among treatment, baseline demographic information, survey responses, and mobile health signals.  

A common concern with network analysis in psychology is replicability~\citep{Maxwell2015, Fried2017, Forbes2019}. 
Some have commented on the instability of network methods~\citep{Borsboom2017}, while others have argued the instability is caused by the use of single-item assessments and small samples~\citep{deyoung2018}. 
To promote robustness, methods have been developed for evaluating the precision and stability of estimated  parameters~\citep{Epskamp2018tutorial}.
In this paper, we identify and address one such relevant cause, called a silent killer of replicability \citep{benjamini2020selective}, when selection bias from the estimated conditional dependence relationships in graphical models is simply ignored during inference.
Several methods have been adopted in the empirical network literature to estimate conditional dependence relationships and attach point estimates to such relationships. 
Nonetheless, the issue of replicability is still a subject of ongoing debate~\citep{Forbes2019}, especially because psychology continues to grapple with a replication crisis, as noted by \cite{Maxwell2015}.

While modeling the conditional dependence relationships in multivariate mHealth data is an indispensable goal, given the risk of false discoveries and growing concerns of replicability, reporting findings from the estimated graph can be grossly misleading without accompanying uncertainty estimates for the matched parameters.  
Recognized as a problem of ``post-selection inference" or ``selective inference" \citep{benjamini2005false, Berk2013, lee2016exact}, countering selection bias from the estimation of such conditional dependence relationships remains a key, unaddressed challenge.

In this paper, we propose a selective inference approach to quantifying the uncertainty in the estimated graph for mHealth data. 
Our focus is on the Gaussian graphical model, which associates an undirected graph to $p$ jointly normally distributed random variables. 
The nodes of the graph represent these variables, while the edges between the nodes capture their conditional dependence relationships. 
These relationships are characterized by the nonzero entries of the inverse covariance matrix, also known as the precision matrix. 
In short, we will refer to this model as GGM and provide a brief overview of it in the next section.

The rest of our paper is organized as follows.
In Section \ref{sec:bg}, we review the GGM and the neighborhood selection approach, which is a multivariate regression method used for estimating conditional dependence relationships. In the same section, we discuss the contributions of our selective inference method. 
In Section \ref{sec:method}, we present our method to address feasible selective inferences for the estimated edges in the GGM. 
We provide an efficient algorithm that can be used to numerically compute a pivot for selective inference, which can produce p-values and confidence intervals for the matched parameters.
In Section \ref{sec:sim}, we present the results of simulations that investigate the performance of our method in different settings. 
In Section \ref{sec:prompt}, we discuss the findings from applying our method to the PROMPT mHealth trial.
Finally, we conclude our paper in Section \ref{sec:con} with a brief summary.


\section{Background and contributions}
\label{sec:bg}

\subsection{Gaussian graphical model}

We start by briefly reviewing the GGM.
Let $X$ be a $p$ dimensional random vector where
$$X = \left(X^{[1]}, X^{[2]}, \ldots, X^{[p]}\right) \sim N_p(0_p, \Sigma),$$
and $\Sigma$ is an invertible $p\times p$ covariance matrix. 
Let 
$\Theta = \Sigma^{-1}$
be the precision matrix, with the $(j,k)^{\text{th}}$ element denoted by $\theta_{j,k}$.

Denote by $X^{[j]}$ the $j^{\text{th}}$ entry in the $p$ dimensional Gaussian vector $X$.
If $j \neq k \in \{1,2,\ldots, p\}$, then $X^{[j]}$ and $X^{[k]}$ are independent conditional on the remaining entries in the random vector $X$ if and only if $\theta_{jk} = \theta_{kj} = 0$.
That is,
$$X^{[j]} \indep X^{[k]} \Big\lvert 
	\begin{pmatrix}
		X^{[l]}
	\end{pmatrix}_{l\neq j,k}
	\text{ if and only if } \theta_{jk}=\theta_{kj} = 0.$$
This fact implies that the support set $\text{Supp}(\Theta)$, excluding the diagonal entries, represents the edge structure of the graph.

Moreover, the $(j,k)^{\text{th}}$ entry of the precision matrix $\Theta$ is, up to a positive scalar, the regression coefficient of the $k^{\text{th}}$ variable in the multiple regression of $j^{\text{th}}$ variable on the rest, and vice-versa.
This forms the basis of multiple regression, a popular framework for estimating the conditional dependence relationships in a GGM.

\subsection{Selecting edges in the GGM}

Suppose that we observe $n$ independent realizations of $X$, which we denote by
$$X_1, X_2, \ldots, X_n,$$ 
and let $\bX$ be the $n \times p$ matrix with $X_i \in \mathbb{R}^p$ in its $i^{\text{th}}$ row.
Let $\bX^{[-j]}$ represent the submatrix of $\bX$ excluding the $j^{\text{th}}$ column and let $\bX^{[j]}$ represent the $j^{\text{th}}$ column of $\bX$.

Neighborhood selection, introduced by \cite{meinshausen2006high} computes a series of nodewise Lasso regressions by solving 
\begin{align}
\label{eq:originalNBD}
\underset{b \in \mathbb{R}^{p-1}}{\text{minimize}}\ \ 
\frac{1}{2} \| \bX^{[i]} - \bX^{[-i]} b \|_2^2 + \lambda_i \| b \|_1
\end{align} 
for $i\in \{1,2,\ldots, p\}$.
Note, the regression at node $i$ uses the $i^{\text{th}}$ variable as the response and the remaining variables as predictors. 

The nodewise Lasso coefficients estimate the neighborhood for each of the $p$ variables, and are combined to estimate $\text{Supp}(\Theta)$, which is equivalent to estimating the edge structure in the graph. 
For example, the $(j,k)^{\text{th}}$ entry of $\Theta$ is estimated to be non-zero using the ``OR'' logic if either the estimated Lasso coefficient of the $j^{\text{th}}$ variable on the $k^{\text{th}}$ variable or the estimated coefficient of the $k^{\text{th}}$ variable on the $j^{\text{th}}$ variable is non-zero. 
Alternatively, the ``AND'' logic can be used to combine the Lasso coefficients for estimating the graph.
The nodewise regression approach taken by neighborhood selection has other important extensions, including a symmetric Lasso regression approach by \cite{peng2009partial}, as well as the estimation of more general network models beyond GGM \citep{ravikumar2010high, yang2015graphical}.

An alternative method for estimating a sparse precision matrix is the graphical lasso by \cite{friedman2008sparse}.
It uses the maximum likelihood approach to produce a matrix estimator of $\Theta$, which can be solved through a block-wise coordinate descent algorithm. 
Although this method is more computationally intensive than neighborhood selection, which solves $p$ separable lasso regressions, it may be preferred in applications where an estimator for entries of the precision matrix is desired, rather than just a sparse support set for the edges.
On the other hand, if the main goal of analysis is to estimate the sparse edge set, then neighborhood selection is usually easier to solve than graphical lasso. 
In addition, the nodewise Lasso regressions in neighborhood selection can be run in parallel on the data, making it substantially faster than graphical lasso.  
Finally, in many studies, measurements are of many different types (e.g., continuous, binary, counts).  Recent work~\citep{Park2020, pmlr-v33-yang14a} has extended graphical models to handle exponential family distributions in
 heterogeneous domains while permitting fast structure/parameter learning. While yet to be used extensively in practice, these methods provide psychologists with a useful tool for learning network models from heterogeneous data.  The neighborhood selection algorithms we consider can be readily extended to this more general setting.  This makes them a natural candidate as a starting point for selective inference for sparse Gaussian graphical models, as we consider the more general exponential family setting an important future direction.

\subsection{Related work and our contributions}

After conducting a series of Lasso regressions in equation \eqref{eq:originalNBD}, we obtain an estimate for the support of the precision matrix, denoted by $\text{Supp}(\widehat{\Theta})$. 
However, to obtain valid p-values or confidence bounds for the entries of $\Theta$ in the estimated support set using the same available data that was utilized in neighborhood selection, it is important to adjust for the selection of edges. 

In Figure \ref{fig:motivating}, ``Naive'' confidence intervals for the nonzero entries of the precision matrix, in $\text{Supp}(\widehat{\Theta})$, are constructed for a range of values of $\lambda_i=\lambda$, $i\in \{1,2,\ldots, p\}$ when data is generated from a GGM and nonzero edges are estimated using neighborhood selection.
These intervals do not take into account the effect of selection. 
The mean coverage rate of the ``Naive'' intervals falls much below the target $90\%$ rate, which highlights the pitfalls of not adjusting for the selection of edges. 

An alternative approach is to use a subsample of the data for neighborhood selection and the remaining holdout samples that were not used for selection to construct intervals. 
This method, commonly referred to as ``Data Splitting'', is depicted in the same Figure. 
While ``Data Splitting'' provides valid selective inferences, it is wasteful since the resulting intervals are only based on the holdout dataset.

In this paper, we propose novel methodology that makes use of leftover information from selection to produce valid inferences for the selected entries in $\Theta$.
This method, which is presented in the next section, is referred to as ``Proposed'' in Figure \ref{fig:motivating}. 
Our proposal not only achieves the prespecified coverage rate across the entire range of $\lambda$ but also produces shorter intervals than data splitting.
Both ``Data Splitting'' and ``Proposed'' generate longer intervals than ``Naive'', as they should, to account for the selection of edges. 
Below we discuss our key contributions  and provide an overview of related work in selective inference.

\begin{figure}[h]
\centering
\includegraphics[scale=0.55]{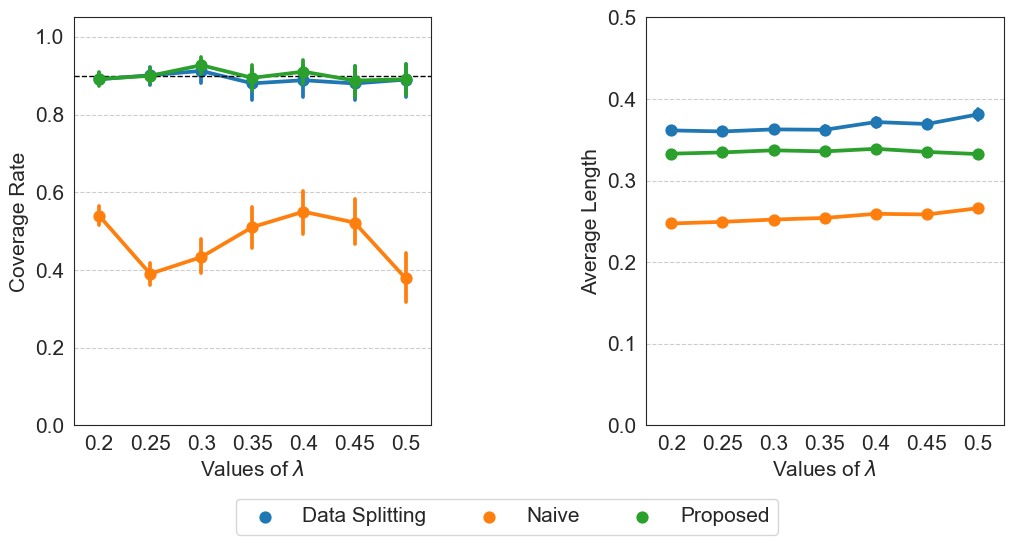}
\caption{Coverage Rates and Average Confidence Interval Lengths of ``Naive'', ``Data Splitting'', and the Proposed Method for Gaussian Data with $n=200$, $p=10$ over 200 Simulations}
\label{fig:motivating}
\end{figure}

Several methods have been proposed to address selective inference after Lasso regression.
These include simultaneous methods by \cite{BerkBrownBujaZhangZhao2013, BachocPreinerstorferSteinberger2016}, conditional methods by \cite{FithianSunTaylor2014, LeeSunSunTaylor2016}, and randomized conditional methods by \cite{tian2018selective, panigrahi2022approximate, RacinesYoung2022}.
The focus of this research is selective inference on the unknown mean parameter while treating the covariance of the response as a fixed parameter. 
The use of conditional methods in combination with externally added randomization adjusts for the specific selection via Lasso, producing bounded intervals and yielding more powerful inferences compared to those without randomization and conditioning. 
See, for example, results presented in recent papers by \cite{KivaranovicLeeb2021, panigrahi2021selection, panigrahi2022exact}.
Whether or not randomization is applied, a key aspect of the conditional methods used for inference after Lasso regression is a neat polyhedral representation of the selection event. 
In the case of normal data, this means that the selection event can be expressed as a set of linear inequalities involving the normally distributed sufficient statistics, and the additional randomization variables if randomization is used.

Although neighborhood selection involves solving $p$ separable Lasso regressions, selective inference in the GGM is not a simple extension of the method used for usual Lasso regression with normal data. 
This is because the selection event in terms of sufficient statistics in the GGM, which follow a Wishart distribution, does not separate into individual nodewise regressions. 
Furthermore, this event no longer has a polyhedral representation in these sufficient statistics, which makes it even more challenging to provide selective inference for the selected edge parameters.

To tackle both challenges, we introduce a randomized conditional method that explicitly adjusts for selection via neighborhood selection in the GGM. 
Motivated by recent ideas of using randomization to account for non-polyhedral selection events, such as \cite{panigrahi2022selective, panigrahi2023approximate, huang2023selective},  our method adds Gaussian randomization variables to the nodewise Lasso regressions in \eqref{eq:originalNBD}. 
We defer the specific form of Gaussian randomization used for our problem to the next section.
Through the use of external randomization, we achieve two important goals.
Firstly, we are able to describe the complex selection event as a set of simple sign constraints.
Secondly, these constraints decouple across all the $p$ regressions, allowing us to 
construct a pivot to infer for the selected entries in the precision matrix.
The simplified selection event, facilitated by randomization, enables us to provide 
an efficient numerical algorithm for computing selective inferences in the GGM. 
The primary computing step of our algorithm can be performed in parallel across the nodewise regressions, similar to the selection step, which allows for fast inferences.

There have been recent proposals, such as those by \cite{RacinesYoung2022} and \cite{neufeld2023data}, suggesting alternative forms of randomization that differ from adding a randomization term to the regression objective. 
These methods can be used to create two randomized, independent copies of data for selection and inference, similar to data splitting, for normal data with a known covariance matrix. However, it is important to note that these methods do not provide a way to split the data into independent copies when the covariance is unknown, which is the case with the GGM. 
Our specific form of randomization provides a way to carry over leftover information at the time of selecting the edges to drawing selective inferences.

\section{Method}
\label{sec:method}
\subsection{Randomized neighborhood selection}

In this section, we present our randomized method for estimating the support set of $\Theta$ and addressing selective inferences for the estimated edges in the GGM. 
As emphasized earlier, the use of randomization is crucial for achieving feasible selective inferences in this problem.

We start with the randomized neighborhood selection method, which gives us the estimated support set for $\Theta$.
Consider $p$ Gaussian randomization variables 
$\wvar{i}$ drawn from $N_{p-1}(0, \Omega^{[i]})$
for $i\in \{1,2,\ldots,p\}$, where $\wvar{i}$ is independent of $\bX$ and $\wvar{i}$ is independent of $\wvar{j}$ for all $i\neq j$.  
We solve $p$ nodewise Lasso regressions as
\begin{equation}
\label{eq:nbd_selection}
\underset{b \in \mathbb{R}^{p-1}}{\text{minimize}} \ \ \left \{
\frac{1}{2} \| \bX^{[i]} - \bX^{[-i]} b \|_2^2 + \lambda_i \| b \|_1 + \frac{\epsilon}{2} \| b \|_2^2  - b^\top \wvar{i}
\right \},
\end{equation}
for each $i \in \{1,2,\ldots,p\}$.
The optimization problem is made strongly convex by including an additional ridge penalty with tuning parameter $\epsilon \in \mathbb{R}^{+}$. 
In our practical implementations of the method, we set the value of $\epsilon$ to a small positive value.
The single regression problem with an added Gaussian randomization variable in the objective, described in \cite{harris2016selective}, is known as the randomized LASSO.

Adding a Gaussian randomization variable to each nodewise regression introduces a tradeoff between selection and inference, allowing us to reserve some information from the selection step to perform selective inference. 
During inference, our method conditions on the event of selection, which is based on the solution to the randomized neighborhood selection. 
This allows us to adjust for the selection of edges while using leftover information from data used during this step. 
Note that our method differs from data splitting, which only uses holdout data for inference. 
Figure \ref{fig:motivating} confirms the expected gain in power over data splitting, providing a preview of our method's performance.

The solution of \eqref{eq:nbd_selection} leads us to observe 
$$\E{i}\subset\left(\{1,2,\ldots, p\}\setminus \{i\}\right),$$
the nonzero entries of the Lasso coefficients from the $i^{\text{th}}$ regression.
Put another way, the set $\E{i}$ gives us the estimated neighborhood for the $i^{\text{th}}$ variable.
We define $\q{i}$ as the number of non-zero entries in $\E{i}$, $\bE{i}$ as the complement set of $\E{i} \cup \{i\}$ and $\bq{i}$ as the size of $\bE{i}$, i.e., $|\E{i}|= \q{i}$, $\bE{i}= \left(\E{i} \cup \{i\}\right)^c$, and $|\bE{i}|= p-1-\q{i}=\bq{i}$.

By combining the estimated neighborhoods of the $p$ nodes using either the ``AND'' or ``OR'' rule, we obtain an estimate for the support of the precision matrix.
In our next step, we present our method for drawing selective inference for $\theta_{j,k}$ whenever $(j,k)$ is an entry in this estimated support set.
However, before we delve into selective inference in the GGM, we fix some basic notations to warm up.

\subsection{Some basics}

Let $M$ be a matrix in $\mathbb{R}^{p\times q}$ and let $A$ and $B$ be subsets of $\{1,2, \ldots, p\}$ and $\{1,2,\ldots, q\}$, respectively. 
We denote the submatrix of $M$ with rows from $A$ and columns from $B$ as $M_{A,B}$. 
Moreover, we denote the submatrix of $M$ with columns from $B$ as $M_B$.
Similarly, if $V$ is a vector in $\mathbb{R}^p$, we denote the subvector of $V$ with components from $A$ as $V_A$. 
We use $I_{m,m}$ to represent the identity matrix of dimension $m$, $0_{p,q} \in \mathbb{R}^{p\times q}$ to represent the zero matrix of dimensions $p\times q$, and $0_q \in \mathbb{R}^q$ to represent the zero vector of dimension $q$.
Throughout, we use the symbol $\phi(x;\mu, \Sigma)$ to represent the density at $x$ of a multivariate Gaussian distribution with mean $\mu$ and covariance matrix $\Sigma$.
We use $\wishartl$ for the Wishart density at $s \in \mathbb{R}^{p\times p}$, which is given by
$$
\wishartl \propto(\operatorname{det}s)^{(n-p-1)/2} \exp{\left( -\frac{1}{2} \operatorname{tr}(\Theta s) \right)}\cdot \ind_{\SSp^{p}}(s),
$$
where $\SSp^{p}$ is the cone of $p$-dimensional positive definite matrices.

We start by simplifying some of our notations. To this end, we assume $\lambda_i= \lambda$ for $i\in \{1,2,\ldots, p\}$.
For this purpose, we assume that $\lambda_i=\lambda$ for every $i$ in the set $\{1,2,\ldots, p\}$. 
Note that our method for selective inference can easily be generalized even when we have $p$ distinct tuning parameters. 
Additionally, we assume without a loss of generality that the predictor matrix and randomization variable in each regression are reordered to have the components in $\E{i}$ stacked above the components in $\bE{i}$.

We define some more estimators besides the selected set of edges $E_i$ for $i\in \{1,2,\ldots, p\}$, which will be needed for constructing selective inferences.
Denote by
$$
\begin{pmatrix}
\Svar{i}\\
\Zvar{i}
\end{pmatrix} = \partial {\| b \|_1}\Big\lvert_{\Bvar{i}},
$$
the subgradient of Lasso penalty at the solution of the $i^{\text{th}}$ regression, where 
$$\Svar{i} = \text{sign} \left( \Bvar{i} \right) \text{ and } \Big\|\Zvar{i}\Big\|_{\infty} \leq 1.$$

Let $S= \bX^\top \bX$, with support equal to $\SSp^{p}$. 
We represent the $(j,k)^{\text{th}}$ entry of $S$ as $S_{j,k}$. 
The value of $S_{j,k}$ is calculated by taking the dot product of the $j^{\text{th}}$ and $k^{\text{th}}$ columns of $\bX$. 
In the rest of the paper, we use $s$ to denote the realized value of $S$ and $s_{i,j}$ to denote the realized value of its $(j,k)^{\text{th}}$ entry.

Using these notations, let
\begin{align*}
	\T{i} = -s_{-i, i},\ 
	\U{i} =  \begin{bmatrix} s_{\E{i}, \E{i}} + \epsilon I_{\q{i}, \q{i}} & 0_{\q{i}, \bq{i}} \\ s_{\bE{i}, \E{i}} & \lambda  I_{\bq{i}, \bq{i}} \end{bmatrix}, \ 
	\V{i} = \lambda\begin{pmatrix} \svar{i} \\ 0_{\bq{i}}\end{pmatrix}.
\end{align*}
For fixed $s$, define the mapping $\Cv{i}:\bbR^{p-1} \to \bbR^{p-1}$ as
\begin{equation}
\Cv{i}(b, z) = \T{i} + \U{i} \begin{pmatrix} b \\ z \end{pmatrix} +  \V{i},
\label{CoV}
\end{equation}
where $b\in \bbR^{\q{i}}$ and $z \in \bbR^{\bq{i}}$, and define its inverse function as
$$\iCv{i} = \left(\Cv{i}\right)^{-1}.$$

Let $\cH{i} =\left\{b: \text{Sgn}(b)= \svar{i}\right\}$ and $\cK{i} =\Big\{z: \|z\|_{\infty}\leq 1\Big\}$. 
Then, we note that
\begin{equation}
\renewcommand{\arraystretch}{1.}
\begin{aligned}
\label{eq:KKTstationarity}
\wvar{i} 
&=  -s_{-i, i} + \begin{bmatrix} s_{\E{i}, \E{i}} + \epsilon I_{\q{i}, \q{i}} & 0_{\q{i}, \bq{i}} \\ s_{\bE{i}, \E{i}} & \lambda I_{\bq{i}, \bq{i}} \end{bmatrix} \begin{pmatrix} \bvar{i} \\ \zvar{i} \end{pmatrix} + \lambda\begin{pmatrix} \svar{i} \\ 0_{\bq{i}}\end{pmatrix}\\
&= \T{i} + \U{i} \begin{pmatrix} \bvar{i} \\ \zvar{i} \end{pmatrix} +  \V{i}\\
&= \Cv{i}\left(\bvar{i}, \zvar{i}\right),
\end{aligned}
\end{equation}
where
$$\bvar{i} \in \cH{i} \text{ and }  \zvar{i}\in  \cK{i}.$$

Finally, we specify some of our general notation rules.
We use lowercase letters to represent variable realizations in our data.
For example, the observed realizations for the random variables $\Wvar{i}$, $\Bvar{i}$, $\Svar{i}$, $\Zvar{i}$ are denoted by $\wvar{i}$, $\bvar{i}$, $\svar{i}$, $\zvar{i}$, respectively. Furthermore, we use uppercase bold font letters to denote collections of random variables that are collected for the $p$ nodewise regressions, and lowercase bold font letters to denote their corresponding realized values. 
For example, we use $\BW = \left\{\Wvar{i}\right\}_{i=1}^p$, $\BBvar = \left\{\Bvar{i}\right\}_{i=1}^p$,
	$\BSvar = \left\{\Svar{i}\right\}_{i=1}^p$, $\BZvar = \left\{\Zvar{i}\right\}_{i=1}^p$ to represent collections of the variables $\Wvar{i}$, $\Bvar{i}$, $\Svar{i}$, $\Zvar{i}$ for $i\in \{1,2,\ldots, p\}$, and we use $\mathbf{w}=\left\{\wvar{i}\right\}_{i=1}^p$, $\bbobs = \left\{\bvar{i}\right\}_{i=1}^p$,
	$\bsobs = \left\{\svar{i}\right\}_{i=1}^p$, $\bzobs = \left\{\zvar{i}\right\}_{i=1}^p$ for their respective realizations.
 
We are ready to construct inferences for the selected edges in the estimated support of $\Theta$.

\subsection{An exact adjustment for the selection of edges}

In this section, our main result in Proposition \ref{prop:condl_likelihood} provides an adjustment for the selection of edges in the GGM, based on the solution to \eqref{eq:nbd_selection}.
This adjustment, when multiplied with the Wishart density of the random matrix $S$, gives us the starting point to derive pivots for selective inferences. 

We note that 
\begin{equation}
\Big\{\BE= E, \BSvar = \bsobs\Big\} = \left\{\bvar{i} \in \cH{i} , \zvar{i} \in \cK{i} \text { for } i\in [p] \right\},
\label{cond:event}
\end{equation}
which is a direct result of the well-studied Lasso penalty. 
To obtain an adjustment for the selection of edges, we first derive the joint distribution of $S,\BBvar, \BZvar$ when conditioned on the above-stated selection event.
We present Proposition \ref{prop:cond_lik_C0}, which leads to our main result in the section.

\begin{proposition} 
Consider the event
$$\cC_0= \Big\{\BE= E, \BSvar = \bsobs\Big\},$$
and let
\begin{equation*}
\begin{aligned}
D = \int  f_{\Theta}(\stil) \cdot   \prod_{i=1}^p\phi\left( \Pi_{\stil}^{[i]}\left(\btil{i}, \ztil{i}\right); 0_{p-1}, \Om{i}\right) & \cdot \operatorname{det} \left( \stil_{E_i, E_i} + \epsilon  I_{\q{i}}\right)  \\
&\;\;\;\;\;\;\;\times \ind_{\cH{i}}\left(\btil{i}\right)\cdot  \ind_{\cK{i}}\left(\ztil{i}\right)d\btil{i} d\ztil{i} d\stil.
\end{aligned}
\end{equation*}
Conditional on the event $\cC_0$, the density function of $S,\BBvar, \BZvar$ at $(s, \bbobs, \bzobs)$ is equal to
\begin{equation*}
\begin{aligned}
f_{\Theta; \cC_0}(s, \bbobs, \bzobs) &=D^{-1}\cdot  f_{\Theta}(s) \cdot \prod_{i=1}^p \Big\{ \phi\left(\Cv{i} \left( \bvar{i}, \zvar{i} \right); 0_{p-1}, \Omega^{[i]} \right) \operatorname{det} \left(s_{\E{i},\E{i}} + \epsilon I_{\q{i}} \right)\\
	&\;\;\;\;\;\;\;\;\;\;\;\;\;\;\;\;\;\;\;\;\;\;\;\;\;\;\;\;\;\;\;\;\;\;\;\times \ind_{\cH{i}}\left(\bvar{i}\right) )\cdot  \ind_{\cK{i}}\left(\zvar{i}\right) \Big\}.
\end{aligned}
\end{equation*}
\label{prop:cond_lik_C0}
\end{proposition}

\begin{proof}
Observe that the joint density of the randomization variables given $S=s$ is
\begin{align*}
f_{\BW|S}\left(\wvar{1}, \dots, \wvar{p} \mid s\right) = \prod_{i=1}^p \phi\left(\wvar{i}; 0_{p-1}, \Omega^{[i]} \right).
\end{align*}
For fixed $s$, applying the change of variables 
$$\Wvar{i} \stackrel{\iCv{i}}{\mapsto} \left(\Bvar{i}, \Zvar{i} \right),$$
for each $i \in [p]$, gives us
\begin{align*}
f_{\BBvar, \BZvar|S}(\bbobs, \bzobs \mid s) \propto 
 \prod_{i=1}^p \left(\det D_{\Cv{i}}\left(\bvar{i}, \zvar{i}\right)\right) \cdot \phi\left(\Cv{i} \left( \bvar{i}, \zvar{i} \right); 0_{p-1}, \Omega^{[i]} \right),
\end{align*}
where $\det D_{\Cv{i}}\left(\bvar{i}, \zvar{i}\right)=\lambda^{\bq{i}}\operatorname{det} \left(s_{\E{i},\E{i}} + \epsilon I \right)$ is the Jacobian matrix associated with the mapping ${\Cv{i}}$, computed in Lemma \ref{lem:Jacobian}.

Clearly, the joint density of $(S, \BBvar, \BZvar)$ is equal to
\begin{equation*}
\begin{aligned}
f_{S, \BBvar, \BZvar}(s, \bbobs, \bzobs) &= \wishartl \cdot f_{\BBvar, \BZvar|S}(\bbobs, \bzobs \mid s),
\end{aligned}
\end{equation*}
when combined with the marginal Wishart density of $S$.
Because of the equivalence in \eqref{cond:event}, the density function of $(S, \BBvar, \BZvar)$ conditional on $\cC_0$ is given by
\begin{equation*}
\begin{aligned}
f_{\Theta; \cC_0}(s, \bbobs, \bzobs)  &\propto \wishartl \cdot f_{\BBvar, \BZvar|S}(\bbobs, \bzobs \mid s) \cdot \ind_{\cH{i}}\left(\bvar{i}\right) \cdot  \ind_{\cK{i}}\left(\zvar{i}\right) \\
&\propto \wishartl \prod_{i=1}^p \Big\{ \phi\left(\Cv{i} \left( \bvar{i}, \zvar{i} \right); 0_{p-1}, \Omega^{[i]} \right) \operatorname{det} \left(s_{\E{i},\E{i}} + \epsilon I_{\q{i}} \right)\\
& \;\;\;\;\;\;\;\;\;\;\;\;\;\;\;\; \;\;\;\;\;\;\;\;\;\;\;\;\;\;\;\;  \;\;\;\;\;\;\;\;\;\;\;\;\;\;\;\;  \;\;\;\;\;\;\;\;\;\;\;\;\;\;\;\;  \;\;\;\;\;\;\times \ind_{\cH{i}}\left(\bvar{i}\right) \cdot  \ind_{\cK{i}}\left(\zvar{i}\right) \Big\}.
\end{aligned}
\end{equation*}
Normalizing this density proves our claim.
\end{proof}

\begin{proposition}
Conditional on 
\begin{align}
\label{eq:selectionevent}
	\cC= \left\{\BE= E, \BSvar = \bsobs,  \BZvar = \bzobs\right\},
\end{align}
the density function of $S$ at $s$, denoted by $f_{\Theta; \cC}(s)$, is proportional to
$$ f_{\Theta}(s) \cdot  \prod_{i=1}^p\int_{\cH{i}} \phi\left( \Cv{i}\left(\btil{i}, \zvar{i}\right); 0_{p-1}, \Om{i}\right)\cdot \operatorname{det} \left( s_{E_i, E_i} + \epsilon  I_{\q{i}} \right) d\btil{i} .$$
\label{prop:condl_likelihood}
\end{proposition}

\begin{proof}
Starting from the joint density in Proposition \ref{prop:cond_lik_C0}, we condition further on $\BZvar = \bzobs$.
This gives us the density function of  $S$, $\BBvar$ when conditioned on $\cC$, which is proportional to
$$ f_{\Theta}(s) \cdot  \prod_{i=1}^p \phi\left( \Cv{i}\left(\bvar{i}, \zvar{i}\right); 0_{p-1}, \Om{i}\right)\cdot \operatorname{det} \left( s_{E_i, E_i} + \epsilon  I_{\q{i}} \right) \cdot  \ind_{\cH{i}}\left(\bvar{i}\right) .$$
Marginalizing over $\Bvar{i}$ results in the density of $S$ given $\cC$ and completes our proof.
\end{proof}

To sum up, Proposition \ref{prop:condl_likelihood} provides an exact adjustment to the Wishart distribution of $S$ to account for the selection of edges in the GGM.

\subsection{Pivot for selective inference}

Using the adjusted density from the preceding section, we construct a pivot for each parameter $\theta_{j_0,k_0} \in \Theta^E$, which serves as the main object of selective inference.

To form this pivot, we derive a one-dimensional density function that only involves the parameter we are interested in, by conditioning further on the observed values of 
$$
\bar{S}_{j_0,k_0}= \left\{ S_{j,k}, \ \text{ for all } (j,k) \not = (j_0, k_0) \right\}.
$$
Theorem \ref{thm:condl_likelihood} states this density function, which leads us to our pivot by applying the probability integral transform based on the related cumulative distribution function (CDF).

We define some more notations to proceed.
Let $\bar{s}_{j_0,k_0}$ denote the observed values of the variables $\bar{S}_{j_0,k_0}$.
For $c\in \bbR$, we define
$$
\Ijkfunc{c, \bar{s}_{j_0,k_0}}: \bbR \to  \bbR^{p\times p}
$$
where the $(j,k)^{\text{th}}$ entry of the matrix-valued mapping is given by
$$
\left[\Ijkfunc{c, \bar{s}_{j_0,k_0}}\right]_{j,k} = \begin{cases}
		c & \text{ if } \; (j,k) = (j_0, k_0) \text{ or } \; (j,k) = (k_0, j_0)\\
		s_{j,k} & \text{ otherwise.}
	\end{cases}
$$
Observe that $\Ijkfunc{c, \bar{s}_{j_0,k_0}}$ simply replaces the $(j_0,k_0)^{\text{th}}$ and $(k_0,j_0)^{\text{th}}$ entry of the observed data matrix with $c$, while keeping all its other entries intact.

Recall that the matrices $\T{i}$ and $\U{i}$, which were defined in \eqref{eq:KKTstationarity}, can be viewed as mappings of the data matrix $s$.
We let
$$
\Tjkfunc{i}{c, \bar{s}_{j_0,k_0}} = \T{i} \circ \Ijkfunc{c, \bar{s}_{j_0,k_0}}, \ \Ujkfunc{i}{c, \bar{s}_{j_0,k_0}}= \U{i} \circ \Ijkfunc{c, \bar{s}_{j_0,k_0}}.
$$
These notations specify how $\T{i}$ and $\U{i}$ depend on the $(j_0,k_0)^{\text{th}}$ entry of $s$.
Define
\begin{align}
\label{eq:1d_CoV}
	\univKKT{i}(c, \bar{s}_{j_0,k_0}, b, z) = 
	\Tjkfunc{i}{c, \bar{s}_{j_0,k_0}} 
	+ \Ujkfunc{i}{c, \bar{s}_{j_0,k_0}} \begin{pmatrix} b \\ z \end{pmatrix} 
	+ \V{i}.
\end{align}
Specifically, we can rewrite  \eqref{eq:KKTstationarity}  as 
\begin{equation*}
\renewcommand{\arraystretch}{1.}
\begin{aligned}
	\wvar{i} = \univKKT{i}(s_{j_0,k_0}, \bar{s}_{j_0,k_0}, \bvar{i}, \zvar{i}).
\end{aligned}
\end{equation*}

\begin{theorem}
Define the sets 
$$\cG =\left\{i \in [p]: \ j_0 \in E_i \wedge k_0 \in E_i \right\}, \ \cF =\left\{i \in [p]: \ j_0 \in E_i \vee k_0 \in E_i \right\}.$$
Let
\begin{equation*}
\begin{aligned}
\Lambda_{j_0,k_0}(c, \bar{s}_{j_0,k_0}) &=\prod_{i\in \cF\cup \{j_0, k_0\}} \int_{\H^i} \phi\left(\univKKT{i}(c, \bar{s}_{j_0,k_0}, \btil{i}, \zvar{i}); 0_{p-1}, \Omega^{[i]} \right) d\btil{i}
\\ 
&\;\;\;\;\;\;\;\;\;\;\;\;\;\;\;\;\;\;\;\;\;\times \prod_{i\in \cG}\operatorname{det} \left( \left[\Ijkfunc{c, \bar{s}_{j_0,k_0}}\right]_{E_i, E_i} + \epsilon I_{\q{i}} \right).
\end{aligned}
\end{equation*}
Then, the conditional density of $S_{j_0,k_0}$ given the event in \eqref{eq:selectionevent} and $\bar{S}_{j_0,k_0}=\bar{s}_{j_0,k_0}$, at $c\in \bbR$, is given by
\begin{equation*}
\begin{aligned}
 \frac{ (\operatorname{det}\Ijkfunc{c, \bar{s}_{j_0,k_0}})^{(n-p-1)/2} \exp \left(  - \theta_{j_0 k_0} c\right) 
	\Lambda_{j_0,k_0}(c, \bar{s}_{j_0,k_0})\cdot\ind_{\SSp}(\Ijkfunc{c, \bar{s}_{j_0,k_0}})}
{\bigintsss 
(\operatorname{det}\Ijkfunc{t, \bar{s}_{j_0,k_0}})^{(n-p-1)/2} \exp \left(  - \theta_{j_0 k_0} t\right) 
	\Lambda_{j_0,k_0}(t, \bar{s}_{j_0,k_0})\cdot\ind_{\SSp}(\Ijkfunc{t, \bar{s}_{j_0,k_0}})dt}.
\end{aligned}	
\end{equation*}
\label{thm:condl_likelihood}
\end{theorem}

The previous result indicates that only a subset of the $p$ regressions contribute to the univariate conditional density of $S_{j_0,k_0}$.
This subset can be much smaller than $p$, especially if the $j_0^{\text{th}}$ node or the $k_0^{\text{th}}$ node are significantly associated with a sparse subset of nodes in the graph. 
The simplification of the univariate density from the adjusted density of the random matrix S is explained in detail in the proof, which we defer to the Appendix.

Denote by $\bF_{\cC; \bar{s}_{j_0,k_0}}(\cdot; \theta_{j_0, k_0})$ the CDF for this density.
Corollary \ref{cor:pivot} presents our pivot, which we obtain after applying the probability integral transform to the CDF.

\begin{corollary}
	Conditional on
	$$\cC= \left\{\BE= E, \BSvar = \bsobs,  \BZvar = \bzobs, \bar{S}_{j_0,k_0} = \bar{s}_{j_0,k_0}\right\},$$
	we have
	\begin{align*}
	\label{eq:pivot}
		\bF_{\cC; \bar{s}_{j_0,k_0}}(S_{j_0, k_0}; \theta_{j_0, k_0}) \sim \operatorname{Uniform}(0,1).
	\end{align*}
	\label{cor:pivot}
\end{corollary}

\subsection{Algorithm for computing selective inference}

Finally, we offer an efficient algorithm to numerically compute selective inference using our pivot. 

Instead of computing the integrals in our pivot, we use a constrained optimization problem to calculate them. 
This approach is based on a Laplace approximation, which was previously employed for selective inferences in \cite{panigrahi2017mcmc, panigrahi2021selection}. 
To put it formally, we have
\begin{align*}
	& \int_{\H^i} \phi\left(\univKKT{i}(c, \bar{s}_{j_0,k_0}, \btil{i}, \zvar{i}); 0_{p-1}, \Omega^{[i]} \right) d\btil{i}\\
	\approx & \exp \left( -  \underset{\btil{i} \in \cH{i}}{\text{minimize}} \frac{1}{2} \left(\univKKT{i}(c, \bar{s}_{j_0,k_0}, \btil{i}, \zvar{i})\right)^\top \left(\Omega^{[i]}\right)^{-1} \left(\univKKT{i}(c, \bar{s}_{j_0,k_0},\btil{i}, \zvar{i})\right)  \right).
\end{align*}

In practice, we solve an unconstrained optimization problem through the use of a barrier function
$$\operatorname{Bar}_{\H^i}\left(\btil{i} \right) =\sum_{j \in [\q{i}]} \log \left(1 + \left(\svar{i} \bobs{i}\right)^{-1}\right).$$
Note that the barrier function enforces the sign constraints on our optimization variables by imposing a larger penalty as the variables move near the boundary of $\cH{i}$.
That is, we solve
\begin{equation}
\begin{aligned}
\laplaceint{i}(c, \bar{s}_{j_0,k_0}) = \underset{\btil{i} \in \mathbb{R}^{\q{i}}}{\text{minimize}} \ 
	 	&\Big\{\frac{1}{2} \left(\univKKT{i}(c, \bar{s}_{j_0,k_0}, \btil{i}, \zvar{i})\right)^\top
	 	\left(\Omega^{[i]}\right)^{-1} 
	 	\left(\univKKT{i}(c, \bar{s}_{j_0,k_0},\btil{i}, \zvar{i})\right) \\
		&+ \operatorname{Bar}_{\cH{i}}(\btil{i})\Big\}
\end{aligned}
\label{eq:approx_integral}
\end{equation}
for $i\in \cF \cup \{j_0,k_0\}$.

Now, we are ready to compute our pivot by substituting $\Lambda_{j_0,k_0}(c, \bar{s}_{j_0,k_0})$ with 
\begin{equation}
\begin{aligned}
\widehat\Lambda_{j_0,k_0}(c, \bar{s}_{j_0,k_0})=\prod_{i\in \cF\cup \{j_0, k_0\}} \exp\left(-\laplaceint{i}(c, \bar{s}_{j_0,k_0})\right)\cdot \prod_{i\in \cG}\operatorname{det} \left( \left[\Ijkfunc{c, \bar{s}_{j_0,k_0}}\right]_{E_i, E_i} + \epsilon I_{\q{i}} \right)
\end{aligned}
\label{eq:approx_lambda}
\end{equation}
 in the conditional density in Theorem \ref{thm:condl_likelihood}.
 Denote by
$$
d(c, \bar{s}_{j_0,k_0}) =
 (\operatorname{det}\Ijkfunc{c, \bar{s}_{j_0,k_0}})^{(n-p-1)/2} \exp \left(  - \theta_{j_0 k_0} c\right) 
	\widehat\Lambda_{j_0,k_0}(c, \bar{s}_{j_0,k_0})\cdot\ind_{\SSp}(\Ijkfunc{c, \bar{s}_{j_0,k_0}}).
$$	
Then, the pivot we compute numerically is equal to
 \begin{align}
 \label{eq:approx_pivot}
 \widehat\bF_{\cC; \bar{s}_{j_0,k_0}}(S_{j_0, k_0}; \theta_{j_0,k_0}) &=\left(\sum_{t\in G} d(c, \bar{s}_{j_0,k_0})\right)^{-1}\cdot {\sum\limits_{t\in G: t \leq S_{j_0, k_0}} d(c, \bar{s}_{j_0,k_0})}.
\end{align}
 Algorithm \ref{algo:pivot} provides the above-outlined steps to obtain our pivot.
 
\begin{algorithm}[h]
\caption{Numerically computing our pivot for $\theta_{j_0 k_0}$}
\label{algo:pivot}
\begin{algorithmic}[]

\State $G \gets \{t_0, t_1, ..., t_{M}: t_i - t_{i-1} = \delta\}$, $\delta > 0$
\Comment Grid with $M$ evenly spaced points with a distance of $\delta$
\For{all \( c \) in grid \( G \)}
       \For{ $i\in \cF\cup \cG$} 
    \State Store $\univKKT{i}(c, \bar{s}_{j_0,k_0}, b, z)$ in \eqref{eq:1d_CoV}. 
    \If{$i \in \cG$}
    \State Store $\operatorname{det} \left( \left[\Ijkfunc{c, \bar{s}_{j_0,k_0}}\right]_{E_i, E_i} + \epsilon I_{\q{i}} \right)$ 
    \EndIf
    \State Solve $\laplaceint{i}(c, \bar{s}_{j_0,k_0})$  in \eqref{eq:approx_integral}
    \State Compute $\widehat\Lambda_{j_0,k_0}(c, \bar{s}_{j_0,k_0})$ in \eqref{eq:approx_lambda}
\EndFor
\EndFor
\State Calculate $\widehat\bF_{\cC; \bar{s}_{j_0,k_0}}(S_{j_0, k_0}; \theta_{j_0,k_0})$ as shown in \eqref{eq:approx_pivot} \Comment Pivot

\Return $\widehat\bF_{\cC; \bar{s}_{j_0,k_0}}(S_{j_0, k_0}; \theta_{j_0,k_0})$
\end{algorithmic}
\end{algorithm}

Using our pivot, we compute two-sided p-values for the null hypothesis $H_0: \ \theta_{j_0, k_0} = \theta_0$
as 
$$2 \min \left(\widehat\bF_{\cC; \bar{s}_{j_0,k_0}}(S_{j_0, k_0}; \theta_{0}), 1 - \widehat\bF_{\cC; \bar{s}_{j_0,k_0}}(S_{j_0, k_0}; \theta_{0})\right).$$

Inverting the test provides the $100(1-\alpha)\%$ confidence interval for $\theta_{j_0,k_0}$ as
\begin{align*}
	\left\{\theta_0 \in \mathbb{R}:\quad \alpha/2 < \widehat\bF_{\cC; \bar{s}_{j_0,k_0}}(S_{j_0, k_0}; \theta_{0}) < 1 - \alpha/2\right\}.
\end{align*}
At last, Lemma \ref{lem:monotone} notes that the above set is indeed an interval.
This can be verified by showing that $\widehat\bF_{\cC; \bar{s}_{j_0,k_0}}(S_{j_0, k_0}; \theta)$ is a monotone function in $\theta$.

\begin{lemma}
	The function $\widehat\bF_{\cC; \bar{s}_{j_0,k_0}}(S_{j_0, k_0}; \theta_{j_0,k_0})$ based on our numerical approximation in \eqref{eq:approx_integral} is monotonically increasing in $\theta_{j_0,k_0}$. 
\label{lem:monotone}
\end{lemma}

We provide a proof for Lemma \ref{lem:monotone} in the Appendix, which is adapted from Lemma A.1 in \cite{lee2016exact}.

Algorithm \ref{algo:CI} outlines the steps to construct confidence intervals using our pivot.

\begin{algorithm}[h]
\caption{Constructing confidence intervals for $\theta_{j_0 k_0}$}
\label{algo:CI}
\begin{algorithmic}[1]
\Require Confidence level $\alpha \in (0,1)$\		  
\State $L \gets \{\theta_0, \theta_1, ..., \theta_{K}: \theta_i - \theta_{i-1} = \delta\}$, $\delta > 0$
\Comment Grid with $K$ evenly spaced points with a distance of $\delta$
\State $\text{CI\_List} \gets \{\}$
\Comment Record $\theta \in L$ that are not rejected
\For{all $\theta$ in $L$}
	\If{$\widehat\bF_{\cC; \bar{s}_{j_0,k_0}}(S_{j_0, k_0}; \theta) \in (\alpha/2, 1-\alpha/2)$} 
	\State $\text{CI\_List} \gets \text{CI\_List} \cup \{\theta\}$
	\EndIf
\EndFor
\State $L_{(j_0,k_0)} \gets \min\ \text{CI\_List}$
\Comment Lower confidence bound
\State $U_{(j_0,k_0)} \gets \max\ \text{CI\_List}$
\Comment Upper confidence bound
\\
\Return $\left(L_{(j_0,k_0)},U_{(j_0,k_0)}\right)$
\end{algorithmic}
\end{algorithm}

\section{Simulation}
\label{sec:sim}
\subsection{Setup}

In this section, we assess the performance of our method in providing selective inference for the estimated edges in the GGM and its potential in recovering the conditional dependence relationships in the related graph.

We generate our data from a multivariate Gaussian distribution with a sparse precision matrix. 
To do so, we closely follow the generative scheme described in \cite{meinshausen2006high}. 

First, we construct the symmetric, sparse precision matrix $\Theta$. 
We introduce two parameters: $m\in \mathbb{N}$ and $c\in (0,1)$. 
These parameters control the number of edges per node (or equivalently the sparsity in the precision matrix) and the signal magnitude.
Now, we consider the following steps.
\begin{enumerate}
\setlength{\itemsep}{1.2pt}
	\item Suppose that each covariate (node) $X_i$ corresponds to a point $p_i \in \mathbb{R}^2$, where $p_i$'s are independent and identically distributed  samples from a $2$-dimensional uniform distribution over the interval $[0,1]\times[0,1]$.
	
	\item Two nodes $X_i$ and $X_j$ are connected ($\theta_{ij} = \theta_{ji} \neq 0$) at random with a probability of
$$
\phi(d(i,j) / \sqrt{p}),
$$
where $d(i,j) = ||p_i - p_j||_2$, and $\phi$ is the standard Gaussian probability density function.
	
 \item The edges are generated according to the probability in step 2, and edges are randomly removed until the graph satisfies the requirement that
$$
\sum_{j\neq i} \mathds{1}\{\theta_{ij} \neq 0\} \leq m.
$$ 

\item For each $i$, $\theta_{ii} = 1$, and for connected nodes $i,j$, $\theta_{ij}$ is sampled randomly from a uniform distribution  
	$$\operatorname{Unif}(0, c/m).$$
\end{enumerate} 
We set $\Sigma = \Theta^{-1}$ and sample $n=400$ observations from a ${p=20}$ dimensional multivariate Gaussian $N_p(\mathbf{0}, \Sigma)$.

To solve the problem of randomized neighborhood selection, we use the penalty weights 
\begin{equation}
	\lambda_i(\alpha) = \kappa_i \cdot 2\sqrt{n} \hat\sigma_i  \bar\Phi^{-1}\left(\alpha / 2p^2\right).
 \label{tuning}
\end{equation}
Here, $\alpha$ is a parameter that controls the probability of falsely including an edge in the estimated model.
It is set at $0.1$.
$\bar\Phi$ is the survival function of the standard Gaussian distribution, and $\hat\sigma_i^2 = \frac{1}{n}(\bX^{[i]})^\top \bX^{[i]} $ is the sample variance for the $i^{\text{th}}$ covariate.
The parameter $\kappa_i$ is a scalar value that can be tuned in experiments. 
This value of tuning parameter is taken from \cite{meinshausen2006high}, where it was shown to recover the true sparse support of the precision matrix with probability with a high probability.
Finally, to estimate the conditional dependence relationships in the graph, we consider both the ``AND'' and ``OR'' logic rules for combining the neighborhoods based on $\BE$.
Our method, which is depicted as ``Proposed" throughout this section, is implemented with simple independent Gaussian randomization variables 
$\wvar{i} \sim N_{p-1}(0, I_{p-1,p-1})$
for $i\in \{1,2,\ldots, p\}$.

\subsection{Metrics}

To summarize the performance of our method, we compute the following metrics. 

Let $\wBE$ denote the selected set of nonzero parameters after applying the ``AND'' and ``OR'' logic rules to combine the estimated neighborhoods based on $\BE$.
Let $\mathbf{E}_0$ denote the true sparse support of the precision matrix.

First, we compute the coverage rate of the confidence intervals $C_{\wBE, (j,k)}$ for $\theta_{jk}$ whenever $(j, k) \in \wBE$:
\begin{align*}
	\text{Coverage Rate} = \frac{\left|\left\{(j,k) \in \wBE: \boldsymbol{\theta}_{jk} \in C_{\wBE, (j,k)}\right\}\right|}{|\wBE|}
\end{align*}
for each round of simulation.
Next, we report the average length of the confidence intervals $C_{\wBE, (j,k)} = (L_{\wBE, (j,k)}, U_{\wBE, (j,k)})$:
\begin{align*}
	\text{Average Length} = 
	\frac
	{\sum_{(j,k) \in \wBE} 
	\left(U_{\wBE, (j,k)} - L_{\wBE, (j,k)}\right)}
	{|\wBE|}.
\end{align*}
To examine the accuracy in estimating the conditional dependence relationships in the graph, we compute $\mathrm{F1}$ scores after conducting selective inferences.
After selective inference, we decide which of the selected edges to include in the graph based on whether the confidence interval covers 0 or not. 
These edges are reported as statistically significant discoveries by our method. 
The $\mathrm{F1}$ score is the harmonic mean of precision and recall and is given by the formula:
$$
\mathrm{F} 1=2 \frac{\text { Precision } \times \text { Recall }}{\text { Precision }+ \text { Recall }}.
$$
Here precision is equal to the proportion of the true edges among the  the reported edges
$$
\text { Precision }=\frac{\left|\mathbf{E}_0 \cap\left\{(j,k) \in \wBE: 0 \notin C_{\wBE, (j,k)}\right\}\right|}{\left|\left\{(j,k) \in \wBE: 0 \notin C_{\wBE, (j,k)}\right\}\right|},$$
while recall is defined as the proportion of the reported edges among the set of true edges
$$
\text { Recall }=\frac{\left|\mathbf{E}_0 \cap \left\{(j,k) \in \wBE: 0 \notin C_{\wBE, (j,k)}\right\}\right|}{\left|\mathbf{E}_0\right|}.
$$


\subsection{Baseline for comparison}

We compare our inferential results to those obtained from a common benchmark method known as data splitting. 
This method involves dividing the data into two independent parts. 
In the selection step, half of our data samples are used to solve the neighborhood selection problem and estimate the edge structure in the graph. 
In the inference step, the remaining half of the data samples are reserved to form confidence intervals for the selected edge parameters in the GGM.
Below, we provide a brief description of a pivot based on data splitting.

Borrowing similar notations as in the preceding section, let ${S}^{(1)}_{j, k}$ and 
$$\bar{S}^{(1)}_{j,k}= \left\{ S^{(1)}_{j',k'}, \ \text{ for all } (j',k') \not = (j, k) \right\}$$ 
denote the sufficient statistics in the standard Wishart density for $j\neq k$, but based on only $50\%$ of the data samples.
The superscript ``$(1)$'' distinguishes these statistics from the ones used in the last section, using the full data.

When inferring for $\theta_{j_0 k_0}$ where the edge $(j_0, k_0)$ is included in our graph from the selection step, we use the distribution of ${S}^{(1)}_{j_0, k_0}$.
To obtain a pivot for the parameter of interest, we condition on all other sufficient statistics, given by $\bar{S}^{(1)}_{j_0,k_0}$, in order to eliminate all nuisance parameters. 

Note, with data splitting, the selected edges can be treated as fixed since inferences are conducted on a new independent dataset.
To be precise, the density of ${S}^{(1)}_{j_0, k_0} $ given $\bar{S}^{(1)}_{j_0,k_0}=\bar{s}^{(1)}_{j_0,k_0}$, when evaluated at $c \in \mathbb{R}$, is given by
\begin{align*}
{p}_{\bar{s}_{j_0,k_0}}(c;\theta_{j_0k_0}) = \dfrac{ \left(\operatorname{det} \Ijkfunc{c, \bar{s}^{(1)}_{j_0,k_0}} \right)^{(n-p-1)/2} \exp \left(  - \theta_{j_0 k_0} c \right)
	\cdot\ind_{\SSp^{p}}\left(\Ijkfunc{c, \bar{s}^{(1)}_{j_0,k_0}}\right)
}
{\bigintss
\left(\operatorname{det} \Ijkfunc{t, \bar{s}^{(1)}_{j_0,k_0}} \right)^{(n-p-1)/2}
\exp \left(  - \theta_{j_0 k_0} t \right)\cdot
\ind_{\SSp^{p}}\left(\Ijkfunc{t, \bar{s}^{(1)}_{j_0,k_0}}\right)
 dt} \ .
\end{align*}	
The derivation of this one-dimensional density follows directly from the Wishart density of 
$$S^{(1)}=\{ {S}^{(1)}_{j_0, k_0}\} \cup \bar{S}^{(1)}_{j_0,k_0},$$ 
after we condition on $\bar{S}^{(1)}_{j_0,k_0}=\bar{s}^{(1)}_{j_0,k_0}$.
A pivot for $\theta_{j_0,k_0}$, using data splitting, is obtained immediately by applying a probability integral transform to this density, which is equal to
\begin{align}
\label{eq:naive_pivot}
		\int_{-\infty}^{S^{(1)}_{j_0,k_0}} {p}_{\bar{s}_{j_0,k_0}}(c;\theta_{j_0k_0}) \ ds  \sim  \operatorname{Uniform}(0,1).
\end{align}
 Inverting the pivot yields p-values and confidence intervals for our parameter in focus.

\subsection{Findings}
\label{subsection:simulation_results}

We summarize findings from $500$ rounds of simulations in two main settings.
In Setting I, we fix the graph connectivity parameter for each node as $m=2$ and vary the signal magnitude parameter $c \in \{0.4, 0.5, 0.6, 0.7, 0.8\}$.
In Setting II, we fix the signal magnitude parameter at $c=1$ and vary the graph connectivity parameter $m\in\{1,2,3,4,5\}$.
For both settings, we present a comparison between our method and the baseline method of data splitting using both the ``AND" and ``OR" rules to combine the selected neighborhoods.

The plots in Figure \ref{fig:signal_AND} and \ref{fig:signal_OR} display error bar plots for the inferential and accuracy metrics in Setting I.
Similarly, the plots in Figure \ref{fig:sparsity_AND} and \ref{fig:sparsity_OR} show error bar plots for the inferential and accuracy metrics in Setting II.
Our method is depicted as ``Proposed" in all of the plots.

\begin{figure}[H]
    \centering
    \includegraphics[scale=0.5]{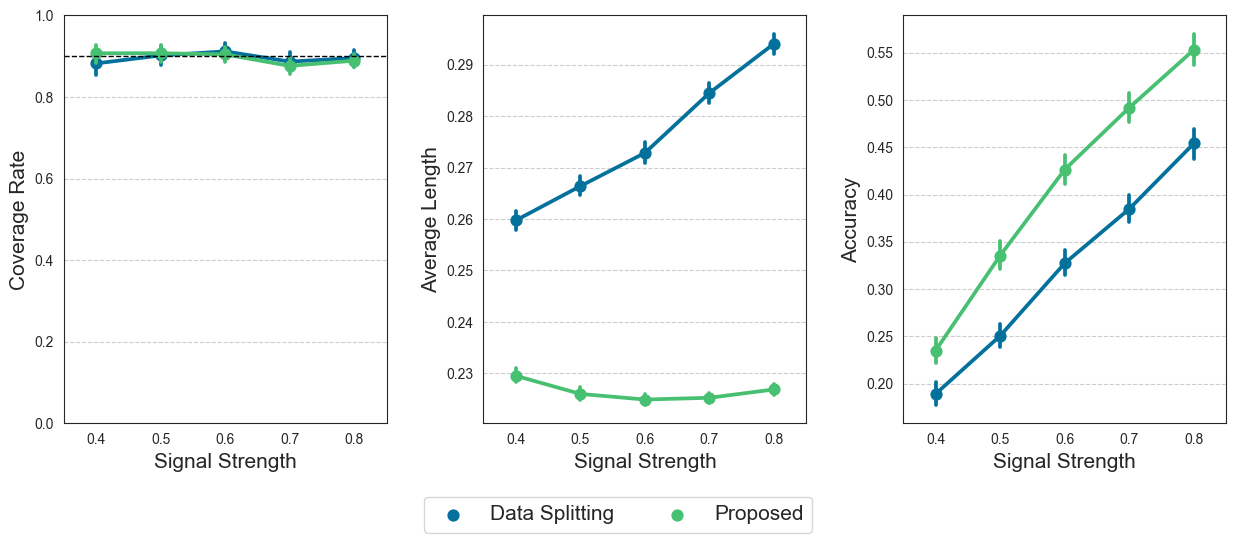}
    \caption{Coverage, length, and accuracy of selective inference for $m=2$ and $c \in \{0.4, 0.5, 0.6, 0.7, 0.8\}$ using the ``AND"  rule.}
    \label{fig:signal_AND}
\end{figure}

\begin{figure}[H]
    \centering
    \includegraphics[scale=0.5]{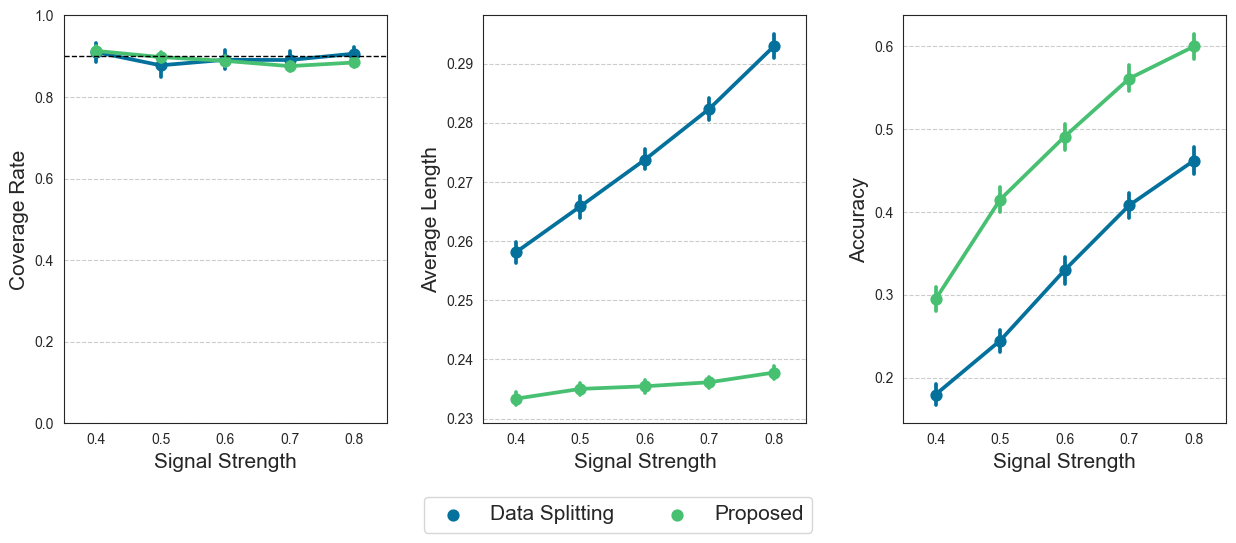}
    \caption{Coverage, length, and accuracy of selective inference for $m=2$ and $c \in \{0.4, 0.5, 0.6, 0.7, 0.8\}$ using the ``OR"  rule.}
    \label{fig:signal_OR}
\end{figure}

\begin{figure}[H]
    \centering
    \includegraphics[scale=0.5]{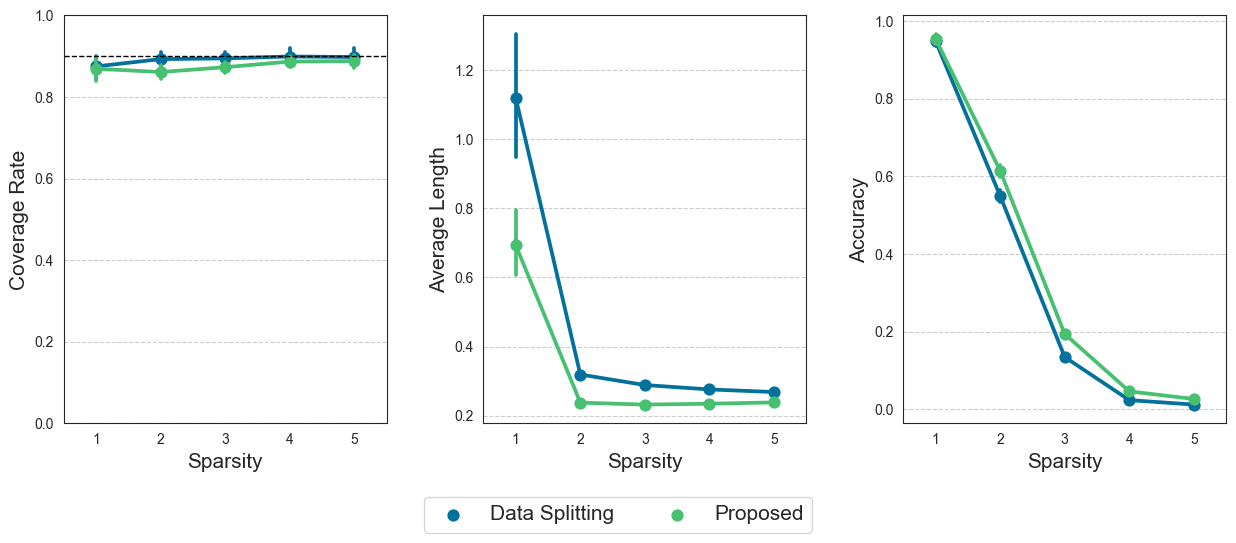}
    \caption{Coverage, length, and accuracy of selective inference for $c = 1$ and $m\in\{1,2,3,4,5\}$ using the ``AND" rule.}
    \label{fig:sparsity_AND}
\end{figure}

\begin{figure}[H]
    \centering
    \includegraphics[scale=0.5]{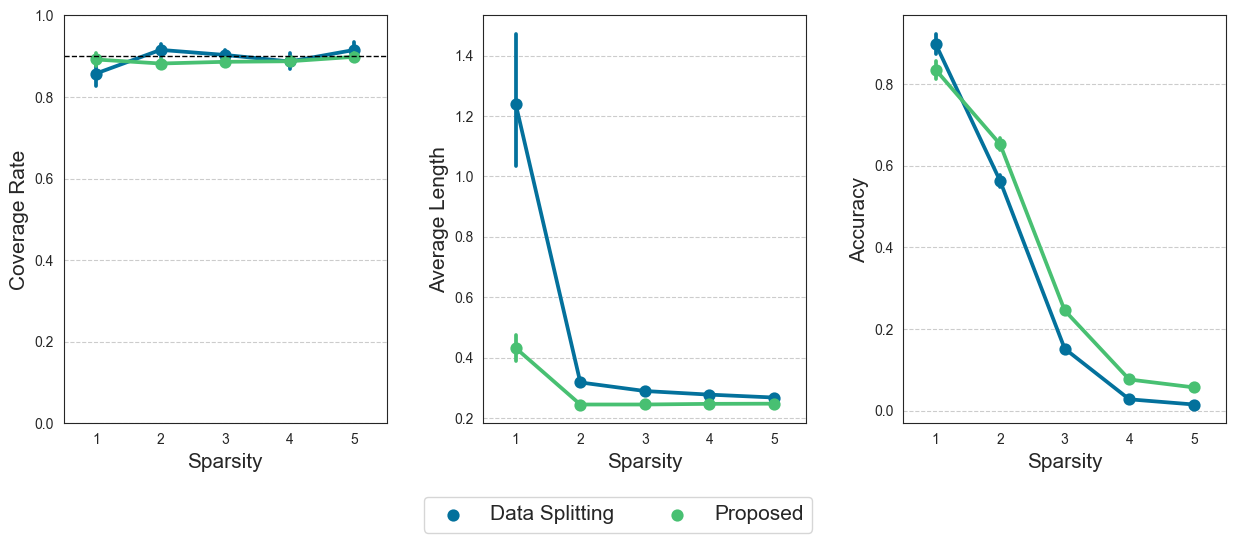}
    \caption{Coverage, length, and accuracy of selective inference for $c = 1$ and $m\in\{1,2,3,4,5\}$ using the ``OR" rule.}
    \label{fig:sparsity_OR}
\end{figure}

For each setting, our simulations show that both ``Proposed" 
and data splitting achieve the target coverage rate, which is set at $90\%$.
However, our ``Proposed" method produces narrower intervals than data splitting, indicating better inferential power. 
This is due to the use of leftover information from the selection step for inference. 
Moreover, our method improves the accuracy of estimating the conditional dependence structure in the graph, which is measured by F1 scores. 
This shows that the ``Proposed" method strikes a better balance between the amount of information used in selecting edges through neighborhood selection and the amount of information used exclusively for selective inference.
All patterns for selective inference and estimation accuracy hold consistently for the ``AND" and ``OR" rules of combining estimated neighborhoods.

\section{Case Study: PROMPT}
\label{sec:prompt}
Depression and anxiety, alongside sleep concerns and addiction, are rapidly escalating global health concerns, leading to increasing disability, lost productivity, and early deaths. Unfortunately, the current healthcare system, relying on traditional face-to-face therapy, is struggling to keep pace with the growing demand for mental health services. The \emph{PROviding Mental health Precision Treatment (PROMPT) Precision Health Study} is a 12-month mobile health intervention trial focused on augmenting standard care to improve health outcomes by using mobile health technologies to extend the reach outside the clinic. Adult patients (age 18+) who have a scheduled adult mental health intake appointment at either Michigan Medicine Outpatient Psychiatry or University Health Service clinics were eligible for participation. Patients were required to have daily access to a smartphone in order to participate. Recruited patients entered study at least 2 weeks prior to their initial clinic appointment.  Patients were randomized to (1) either receive or not receive enhanced feedback (EF) via the study app (e.g., on step count and heart rate goals) and (2) either have access to an additional mental health applicaiton (App) or not. Patients could not be randomized to receive neither EF nor the App, and therefore the study randomized to three conditions (EF + standard of care, App + standard of care, or EF + App + standard of care).  See here for additional details on the PROMPT study.

Participants were tasked with completing surveys throughout the study, including an initial intake survey, surveys at 6 weeks, 18 weeks, and 12 months into their participation. Participants were also notified via the study app on a daily basis to rate their mood on a scale of 1-10. After consenting to the study and completing the intake survey, each study participant received a free Fitbit to wear daily for the duration of their time in the study. 

A key scientific aim of PROMPT is to better understand the complex relationships among treatment, baseline demographic information, survey responses, and mobile health signals.  Here, we focus on the relationship among baseline survey instruments and wearable data collected via the Fitbit between 7 days prior to and 60 days after the baseline survey. This allows us to understand the relationship among these variables prior to the initial clinic visit.   The intake survey included many standard, multi-item questionnaires such as the Patient Health Questionnaire (PHQ-9) and the General Anxiety Disorder (GAD-7).  We pre-processed the intake survey to compute severity scores rather than analyzing individual items.  Second, we summarized the 67 days of Fitbit data which consisted of 15 daily variables.  As some Fitbit data streams require user input, there is a substantial amount of missing data in several data streams. We limit ourselves to datastreams with less than 20\% missing data in each variable. The final list of variables is included in Table \ref{table:varnames} from Appendix \ref{sec:AppendixB}.  We then compute several summary statistics from the remaining daily wearable data such as means and standard deviations. Our final complete dataset consists of $N= 770$ patients with $9$ survey variables, and $15$ sensor variables.

Similar to our simulations in the preceding section, we solve the randomized neighborhood selection with standard Gaussian random variables.
We set the tuning parameters in the nodewise multivariate regressions according to \eqref{tuning}, where $\kappa_i$ is set to 1. 
We applied the ``OR" rule to combine the selected neighborhoods, resulting in the selection of $34$ edges from a total of $276$ possible edges. 
To construct selective inferences for these $34$ edge parameters, we utilized Algorithm \ref{algo:CI}. 

In Figure \ref{fig:graph}, we include a plot for the graph depicting the estimated conditional dependence relationships between the survey and sensor variables.
The solid lines indicate edges which were significant post inference and the dotted lines indicate edges that were included in the graph at the selection step, but were no longer significant post inference.
It is important to note that three of these edge parameters were deemed insignificant by our method, meaning that the confidence intervals returned for these parameters covered $0$.
This is accompanied by Tables \ref{table:sensorCI} and \ref{table:surveyCI}, which report the confidence intervals formed through selective inference.

Our analysis suggests conditional independence between features from the wearable device and the baseline survey items.  We can therefore conclude that there is not a strong relationship between various measures of physical activity and the severity scores among this population of treatment seeking individuals.  We can thus analyze the wearable and baseline survey items separately.  Among the features constructed from Fitbit data, we recover natural relationships such as the conditional dependence between distance and activity calories. Interestingly, calories was found to be conditionally independent of distance given activity calories suggesting that those calories burned during activity windows drives the overall distance covered by the individual.  Among the baseline survey items, we recovered known relationships such as the GAD being conditionally dependent on PHQ and NEO, which measures neuroticism (N), extraversion (E), and openness (O).  Interestingly, the Positive and Negative Suicide Ideation (PANSI) questionnaire was found to be conditionally independent of the GAD but dependent on the NEO and PHQ with the effect seemingly larger for the PHQ. Most importantly, Tables~\ref{table:sensorCI} and~\ref{table:surveyCI} provide uncertainty quantification to ensure replication and reduce the risk of false discoveries from simply applying neighborhood selection to the observed mHealth data.


\begin{figure}[h]
    \centering
    \includegraphics[scale=0.35]{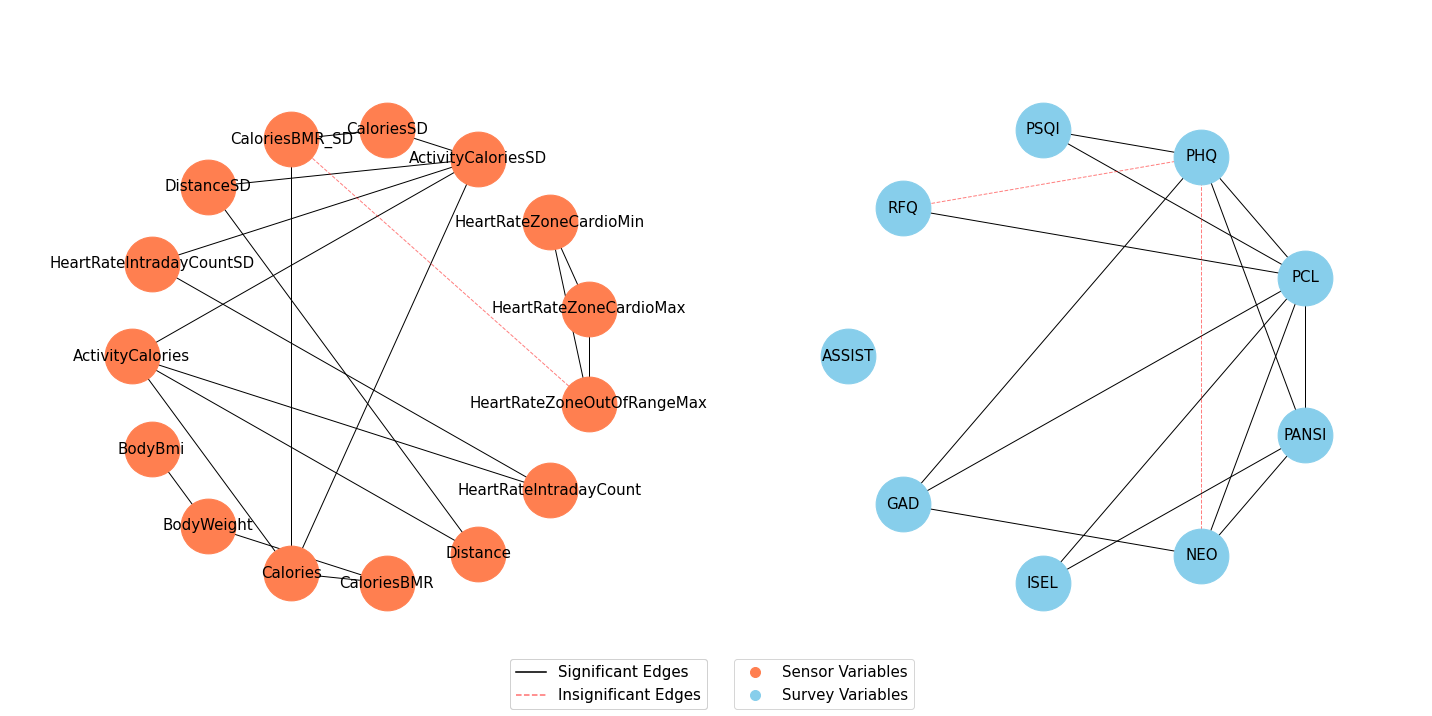}
    \caption{Visualization of graphs depicting the estimated conditional dependence relationships between the survey and sensor Variables}
    \label{fig:graph}
\end{figure}

\begin{sidewaystable}
\tiny
\centering
\begin{tabular}{@{}|c|c|c|c|c|c|c|c|c|c|c|c|c|c|c|c|@{}}
\toprule
& \begin{tabular}[c]{@{}c@{}}Activity\\ Calories\end{tabular} & \begin{tabular}[c]{@{}c@{}}Body\\ Bmi\end{tabular} & \begin{tabular}[c]{@{}c@{}}Body\\
Weight\end{tabular} & Calories 
& \begin{tabular}[c]{@{}c@{}}Calories\\ BMR\end{tabular} & Distance 
& \begin{tabular}[c]{@{}c@{}}HeartRate\\ Intraday\\ Count\end{tabular} & \begin{tabular}[c]{@{}c@{}}HeartRate\\ Zone\\ OutOfRange\\ Max\end{tabular} & \begin{tabular}[c]{@{}c@{}}HeartRate\\ Zone\\ Cardio\\ Max\end{tabular} & \begin{tabular}[c]{@{}c@{}}HeartRate\\ Zone\\ Cardio\\ Min\end{tabular} & \begin{tabular}[c]{@{}c@{}}Activity\\ Calories\\ SD\end{tabular} 
& \begin{tabular}[c]{@{}c@{}}Calories\\ SD\end{tabular}      
& \begin{tabular}[c]{@{}c@{}}Calories\\ BMR\_SD\end{tabular} 
& \begin{tabular}[c]{@{}c@{}}Distance\\ SD\end{tabular}
& \begin{tabular}[c]{@{}c@{}}HeartRate\\ IntradayCount\\ SD\end{tabular} \\ 
\midrule
ActivityCalories
& - & - & - 
& \begin{tabular}[c]{@{}c@{}}(-87.47,\\  -85.56)\end{tabular} 
& - 
& \begin{tabular}[c]{@{}c@{}}(-25.94,\\  -24.03)\end{tabular} 
& \begin{tabular}[c]{@{}c@{}}(-2.89,\\  -1.03)\end{tabular} 
& - & - & - 
& \begin{tabular}[c]{@{}c@{}}(-36.25,\\  -34.34)\end{tabular} 
& - & - & - & - \\ 
\midrule
BodyBmi
& - & - 
& \begin{tabular}[c]{@{}c@{}}(-29.71,\\  -27.80)\end{tabular} 
& - & - & - & - & - & - & - & - & - & - & - & - 
\\ 
\midrule
BodyWeight 
& - & - & - & - 
& \begin{tabular}[c]{@{}c@{}}(-50.65,\\  -48.74)\end{tabular} 
& - & - & - & - & - & - & - & - & - & - \\ 
\midrule
Calories 
& - & - & - & - 
& \begin{tabular}[c]{@{}c@{}}(-97.05,\\  -95.13)\end{tabular} 
& - & - & -  & - & - 
& \begin{tabular}[c]{@{}c@{}}(32.58,\\  34.49)\end{tabular} 
& - 
& \begin{tabular}[c]{@{}c@{}}(36.30,\\  38.21)\end{tabular} 
& - & - \\ 
\midrule
CaloriesBMR 
& - & - & - & - & - & - & - & - & - & - & - & - & - & - & - \\ 
\midrule
Distance 
& - & - & - & - & - & - & - & - & - & - & - & - & - & 
\begin{tabular}[c]{@{}c@{}}(-16.50,\\  -14.59)\end{tabular} & - \\ 
\midrule
\begin{tabular}[c]{@{}c@{}}HeartRate\\ IntradayCount\end{tabular}
& - & - & - & - & - & - & - & - & - & - & - & - & - & - 
& \begin{tabular}[c]{@{}c@{}}(-2.97,\\  -2.39)\end{tabular} \\ 
\midrule
\begin{tabular}[c]{@{}c@{}}HeartRateZone\\ OutOfRangeMax\end{tabular} 
& - & - & - & - & - & - & - & - & 
\begin{tabular}[c]{@{}c@{}}(5.26,\\  6.89)\end{tabular} 
& \begin{tabular}[c]{@{}c@{}}(-10.61,\\  -8.82)\end{tabular}
& - & - & {\textbf{\color[HTML]{CB0000}\begin{tabular}[c]{@{}c@{}}(-1.09,\\  0.20)\end{tabular}}} 
& - & - \\ 
\midrule
\begin{tabular}[c]{@{}c@{}}HeartRateZone\\ CardioMax\end{tabular}     
& - & - & - & - & - & - & - & - & - 
& \begin{tabular}[c]{@{}c@{}}(-50.20,\\  -48.29)\end{tabular} 
& - & - & - & - & - \\ 
\midrule
\begin{tabular}[c]{@{}c@{}}HeartRateZone\\ CardioMin\end{tabular}     
& - & - & - & - & - & - & - & - & - & - & - & - & - & - & -\\ 
\midrule
ActivityCaloriesSD 
& - & - & - & - & - & - & - & - & - & - & - 
& \begin{tabular}[c]{@{}c@{}}(-36.22,\\  -34.31)\end{tabular} 
& - 
& \begin{tabular}[c]{@{}c@{}}(-13.69,\\  -11.78)\end{tabular} 
& \begin{tabular}[c]{@{}c@{}}(-4.01,\\  -2.80)\end{tabular} \\ 
\midrule
CaloriesSD 
& - & - & - & - & - & - & - & - & - & - & - & - & 
\begin{tabular}[c]{@{}c@{}}(-36.87,\\  -34.95)\end{tabular} 
& - & - \\ 
\midrule
CaloriesBMR\_SD 
& - & - & - & - & - & - & - & - & - & - & - & - & - & - & - \\ 
\midrule
DistanceSD 
& - & - & - & - & - & - & - & - & - & - & - & - & - & - & - \\ 
\midrule
\begin{tabular}[c]{@{}c@{}}HeartRate\\ IntradayCountSD\end{tabular}   
& - & - & - & - & - & - & - & - & - & - & - & - & - & - & - \\ 
\bottomrule
\end{tabular}
\caption{Confidence Intervals for Sensor Data}
\label{table:sensorCI}
\end{sidewaystable}

\begin{table}
\centering
\tiny
\begin{tabular}{@{}|c|c|c|c|c|c|c|c|c|c|@{}}
\toprule
       & ASSIST & GAD & ISEL & NEO            & PANSI          & PCL            & PHQ                                  & PSQI           & RFQ                                  \\ \midrule
ASSIST & -      & -   & -    & -              & -              & -              & -                                    & -              & -                                    \\ \midrule
GAD    & -      & -   & -    & (-0.83, -0.54) & -              & (-0.83, -0.51) & (-1.07, -0.73)                       & -              & -                                    \\ \midrule
ISEL   & -      & -   & -    & -              & (-0.36, -0.12) & (-0.33, -0.07) & -                                    & -              & -                                    \\ \midrule
NEO    & -      & -   & -    & -              & (-0.48, -0.22) & (-0.30, -0.01) & {\textbf{\color[HTML]{CB0000}(-0.17, 0.13)}} & -              & -                                    \\ \midrule
PANSI  & -      & -   & -    & -              & -              & (-0.30, -0.01) & (-1.17, -0.86)                       & -              & -                                    \\ \midrule
PCL    & -      & -   & -    & -              & -              & -              & (-0.89, -0.54)                       & (-0.57, -0.30) & (-0.30, -0.04)                       \\ \midrule
PHQ    & -      & -   & -    & -              & -              & -              & -                                    & (-0.34, -0.07) & {\textbf{\color[HTML]{CB0000}(-0.03, 0.25)}} \\ \midrule
PSQI   & -      & -   & -    & -              & -              & -              & -                                    & -              & -                                    \\ \midrule
RFQ    & -      & -   & -    & -              & -              & -              & -                                    & -              & -                                    \\ \bottomrule
\end{tabular}
\caption{Confidence Intervals for Survey Data}
\label{table:surveyCI}
\end{table}

\section{Conclusion}
\label{sec:con}

Precision health studies seek to understand the complex relationships among treatment, baseline demographic information, survey responses, and mobile health signals.
This is achieved by learning the relevant conditional dependence relationships in a graphical model, which is equivalent to the presence or absence of an edge in the related graph.
Although selecting edges and associating point estimates to the selected edges is prevalent, it can be misleading to report these edges without accompanying uncertainty estimates, given growing concerns about replicability. 

In this paper, we propose a method for attaching uncertainties to the selected edges of undirected graphical models by using a selective inference approach. 
Our focus in the paper is on the widely used neighborhood selection method, which estimates the conditional dependence relationships in a graph through nodewise multivariate regressions.
Unlike the usual single regression framework, the selection of edges does not have a simple polyhedral representation. 
However, by utilizing external randomization variables, our method provides an exact adjustment factor to account for the selection of edges. 
This exact adjustment takes the form of a few simple sign constraints, which decouple across the nodewise regressions.

To begin addressing selective inference in undirected graphical models, we considered inference on a single graph from Gaussian data. 
We believe that our current approach will pave the way for other crucial methods to tackle more general models and different types of data.
For instance, this will involve extending the approach to a broader range of graphical models that encompass mixed data types, or performing integrative inferences for graphs that are aggregated across multiple time points or data sources.
We leave these important extensions to future work.

\newpage
\bibliographystyle{apalike}
\bibliography{ref}

\newpage
\bigskip
\begin{center}
{\large\bf SUPPLEMENTARY MATERIAL}
\end{center}

\begin{appendix}
\section{Proofs of Results}

We start with an auxiliary result that we use in the proof of Proposition \ref{prop:cond_lik_C0} for constructing the selection-adjusted density of $S$:
\begin{lemma}[Jacobian of Change of Variable]
\label{lem:Jacobian}
Consider the mapping
\begin{align*}
\Cv{i}\left(\Bvar{i}, \Zvar{i}\right)
&= \T{i} + \U{i} \begin{pmatrix} \Bvar{i} \\ \Zvar{i} \end{pmatrix} +  \V{i} = \Wvar{i}.
\end{align*}
Then, the Jacobian associated with $\Cv{i}$ when viewed as a change of variables mapping from $\Wvar{i}$ to $\left(\Bvar{i}, \Zvar{i} \right)$ is equal to
\begin{align*}
	\lambda_i^{\bq{i}}
	\operatorname{det} \left(s_{\E{i}, \E{i}} + \epsilon I_{\q{i}}  \right).
\end{align*}
\end{lemma}

\begin{proof}[Proof of Lemma \ref{lem:Jacobian}]
Let $D_{\Cv{i}}(\Bvar{i},\Zvar{i})$ be the differential of the map $\Cv{i}$, given by
\begin{align*}
	D_{\Cv{i}}(\Bvar{i},\Zvar{i}) = \frac{\partial\Cv{i}}{\partial \left(\Bvar{i},\Zvar{i}\right)}
\end{align*}
Notice from the definition \eqref{eq:KKTstationarity} that
\begin{align*}
\Cv{i}\left(\Bvar{i}, \Zvar{i}\right)
&= \T{i} + \U{i} \begin{pmatrix} \Bvar{i} \\ \Zvar{i} \end{pmatrix} +  \V{i}
\end{align*}
and therefore
\begin{align*}
	D_{\Cv{i}}(\Bvar{i},\Zvar{i}) = 
	\begin{pmatrix}
		\frac{\partial\Cv{i}}{\partial \Bvar{i}} & 
		\frac{\partial\Cv{i}}{\partial \Zvar{i}}
	\end{pmatrix}
	=
	\begin{pmatrix}
		\left(s_{\E{i}, \E{i}} + \epsilon I_{\q{i}} \right) & 0 \\
		s_{\bE{i}, \E{i}} & \lambda_i I_{\bq{i}}
	\end{pmatrix}.
\end{align*}
Noting that $D_{\Cv{i}}(\Bvar{i},\Zvar{i})$ is a lower triangular matrix, we conclude that
$$
\operatorname{det} D_{\Cv{i}}(\Bvar{i},\Zvar{i}) = \operatorname{det} \left(s_{\E{i}, \E{i}} + \epsilon I_{\q{i}} \right) \det(\lambda_i I_{\bq{i}}) = \lambda_i^{\bq{i}}
	\operatorname{det} \left(s_{\E{i}, \E{i}} + \epsilon I_{\q{i}}  \right).
$$
\end{proof}

\begin{proof}[Proof of Theorem \ref{thm:condl_likelihood}]
Observe that the conditional distribution of $S$ has density proportional to
\begin{align*}
&f_{\Theta}(s) \cdot  \prod_{i=1}^p\int_{\cH{i}}\phi\left( \Cv{i}\left(\btil{i}, \zvar{i}\right); 0_{p-1}, \Om{i}\right)\cdot \operatorname{det} \left( s_{E_i, E_i} + \epsilon  I_{\q{i}} \right) d\btil{i}\\
&\propto (\operatorname{det} s)^{(n-p-1)/2} \cdot \exp \left(- \sum_{j,k \in [p]} \theta_{jk} s_{jk}  \right)\cdot \ind_{\SSp^{p}}(s)\\
&\quad\times \prod_{i=1}^p\int_{\cH{i}}\phi\left( \Cv{i}\left(\btil{i}, \zvar{i}\right); 0_{p-1}, \Om{i}\right)\cdot \operatorname{det} \left( s_{E_i, E_i} + \epsilon  I_{\q{i}} \right) d\btil{i}.
\end{align*}
If we condition further on $\{\bar{S}_{j_0,k_0}=\bar{s}_{j_0,k_0}\}$, then the conditional density of $S_{j_0, k_0}$ at $c$ is proportional to
\begin{align*}
& \left(\operatorname{det}\Ijkfunc{c, \bar{s}_{j_0,k_0}}\right)^{(n-p-1)/2} \exp \left(  - \theta_{j_0 k_0} c \right)\cdot \ind_{\SSp^{p}}\left(\Ijkfunc{c, \bar{s}_{j_0,k_0}}\right)\\
&\quad\times \prod_{i=1}^p\int_{\cH{i}}\phi\left(\univKKT{i}(c, \bar{s}_{j_0,k_0}, \btil{i}, \zvar{i}); 0_{p-1}, \Omega^{[i]} \right)\cdot \operatorname{det} \left( \left[\Ijkfunc{c, \bar{s}_{j_0,k_0}}\right]_{\E{i}, \E{i}} + \epsilon I \right) d\btil{i}.
\end{align*}
The density in the above display follows from the definition of $\Ijkfunc{\cdot,\cdot}$ and the observation that conditional on $\bar{S}_{j_0,k_0} = \bar{s}_{j_0,k_0}$, 
\begin{align*}
	\Pi_{\Ijkfunc{c,\bar{s}_{j_0,k_0}}}\left(\bvar{i}, \zvar{i}\right) 
	= \univKKT{i}(c, \bar{s}_{j_0,k_0}, \btil{i}, \zvar{i})
\end{align*}
for a fixed $c \in \mathbb{R}$.
From the above-stated expression, we conclude that this conditional density is given by
\begin{equation*}
\begin{aligned}
\label{eq:exact_pivot_original}
&\frac{ (\operatorname{det} \Ijkfunc{c,\bar{s}_{j_0,k_0}})^{(n-p-1)/2} \exp \left(  - \theta_{j_0 k_0} c \right) \widetilde\Lambda_{j_0,k_0}(c, \bar{s}_{j_0,k_0})\cdot \ind_{\SSp}\left(\Ijkfunc{c, \bar{s}_{j_0,k_0}}\right)}
{\bigintsss 
(\operatorname{det}\Ijkfunc{t, \bar{s}_{j_0,k_0}})^{(n-p-1)/2} \exp \left(  - \theta_{j_0 k_0} t\right) 
	\widetilde\Lambda_{j_0,k_0}(t, \bar{s}_{j_0,k_0})\cdot\ind_{\SSp}(\Ijkfunc{t, \bar{s}_{j_0,k_0}})dt},
\end{aligned}	
\end{equation*}
where
\begin{align*}
	\widetilde\Lambda_{j_0,k_0}(c, \bar{s}_{j_0,k_0}) =
	\Bigg\{
	&\prod_{i=1}^p  \operatorname{det} \left( \left[\Ijkfunc{c,\bar{s}_{j_0,k_0}}\right]_{\E{i}, \E{i}} + \epsilon I \right) 
	\\
	&\times 
	\int_{\cH{i}} \phi\left(\univKKT{i}(c, \bar{s}_{j_0,k_0}, \btil{i}, \zvar{i}); 0_{p-1}, \Omega^{[i]} \right) d\btil{i} \Bigg\}.
\end{align*}

In the remaining part of the proof, we derive a simplified expression for the conditional density of $S_{j_0,k_0}$. 
In order to obtain this expression, we observe that any term in $\widetilde\Lambda_{j_0,k_0}(c, \bar{s}_{j_0,k_0})$ that does not depend on $c$ will produce a constant factor in our density function.

We consider three cases for each $i \in [p]$, of which CASE II and CASE III, as defined later, give rise to an integral or a determinant term involved in $\widetilde\Lambda_{j_0,k_0}(c, \bar{s}_{j_0,k_0})$ that depend on $c$. 
We then derive a simplified expression for the conditional density of $S_{j_0,k_0}$ using these terms.
\smallskip

\begin{description}
\setlength{\itemsep}{1.5em}
\item[CASE I.] $\Big\{j_0 \notin \E{i} \text{ and } k_0 \notin \E{i}\Big\}$. 
\smallskip

Note that in this case, the matrices $\T{i}$, $\U{i}$, and hence $\Pi_{\Ijkfunc{c,\bar{s}_{j_0,k_0}}}\left(\bvar{i}, \zvar{i}\right)$ does not depend on $s_{j_0, k_0}=c$. Therefore, these terms can be disregarded from the conditional density.

\item[CASE II.] $\Big\{j_0 \in \E{i} \text{ and } k_0 \not \in \E{i}\Big\}$ or $\Big\{j_0 \not \in \E{i} \text{ and } k_0 \in \E{i}\Big\}$.
\begin{figure}[h] 
  \centering
  \begin{tikzpicture}[main/.style = {draw, circle}] 
    \node[main] (1) {$i$};
    \node[main] (2) [below left of=1] {$j_0$};
    \node[main] (3) [below right of=1] {$k_0$};
    \draw (1) -- (2);
    \end{tikzpicture} 
    \tikzset{minimum size=0.75cm}
    \begin{tikzpicture}[main/.style = {draw, circle}] 
    \node[main] (1) {$i$};
    \node[main] (2) [below left of=1] {$j_0$};
    \node[main] (3) [below right of=1] {$k_0$};
    \draw (1) -- (3) ;
    \end{tikzpicture}
    \caption{CASE II}
    \label{fig:int_flag}
\textit{An edge between $i$ and $j_0$ or $k_0$ indicates $j_0 \in \E{i}$ or $ k_0 \in \E{i}$, respectively}
\end{figure}
\smallskip

Recall that
\begin{align*}
	\univKKT{i}(c, \bar{s}_{j_0,k_0}, b, z) &= 
	\Tjkfunc{i}{c, \bar{s}_{j_0,k_0}} 
	+ \Ujkfunc{i}{c, \bar{s}_{j_0,k_0}} \begin{pmatrix} b \\ z \end{pmatrix} 
	+ \V{i}\\
	&=
	-\left[\Ijkfunc{c,\bar{s}_{j_0,k_0}}\right]_{-i,i}
	+\begin{pmatrix} \left[\Ijkfunc{c,\bar{s}_{j_0,k_0}}\right]_{\E{i}, \E{i}} + \epsilon I_{\q{i}, \q{i}} &\ & 0_{\q{i}, \bq{i}} \\ \left[\Ijkfunc{c,\bar{s}_{j_0,k_0}}\right]_{\bE{i}, \E{i}} &\ & I_{\bq{i}, \bq{i}} \end{pmatrix}
	\begin{pmatrix} b \\ z \end{pmatrix} \\
	&\;\;\;\;\;+ \V{i}.
\end{align*}

First, suppose that $i\in [p]\setminus \{j_0, k_0\}$. 
Then, we note that $\Ujkfunc{i}{c, \bar{s}_{j_0,k_0}}$ depends on $s_{j_0, k_0}=c$, while $\Tjkfunc{i}{c, \bar{s}_{j_0,k_0}} = \left[\Ijkfunc{c,\bar{s}_{j_0,k_0}}\right]_{-i,i}$ does not depend on the value of $s_{j_0, k_0}$.
\smallskip

Next suppose that $i \in \{j_0,k_0\}$.
In this case, $\Tjkfunc{i}{c, \bar{s}_{j_0,k_0}}$ depends on $s_{j_0, k_0}=c$, while $\Ujkfunc{i}{c, \bar{s}_{j_0,k_0}}$ does not depend on the value of $s_{j_0, k_0}$.

\smallskip

In both cases, the determinant
$$\operatorname{det} \left( \left[\Ijkfunc{c,\bar{s}_{j_0,k_0}}\right]_{\E{i}, \E{i}} + \epsilon I \right) $$
does not involve $s_{j_0, k_0}=c$, as $\E{i}$ does not contain both $j_0$ and $k_0$ at the same time.
To sum up, for any such $i$, the contribution to $\widetilde\Lambda_{j_0,k_0}(c, \bar{s}_{j_0,k_0})$ is equal to
$$\int_{\cH{i}}\phi\left(\univKKT{i}(c, \bar{s}_{j_0,k_0}, \btil{i}, \zvar{i}); 0_{p-1}, \Omega^{[i]} \right) d\btil{i}.$$

\item[CASE III.] $\Big\{j_0 \in \E{i} \text{ and } k_0 \in \E{i}\Big\}$.
\smallskip

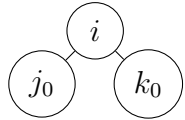
\begin{figure}[h]
\centering
  \tikzset{minimum size=0.75cm}
    \begin{tikzpicture}[main/.style = {draw, circle}] 
    \node[main] (1) {$i$};
    \node[main] (2) [below left of=1] {$j_0$};
    \node[main] (3) [below right of=1] {$k_0$};
    \draw (1) -- (2);
    \draw (1) -- (3);
    \end{tikzpicture} 
    \caption{CASE III} 
    \label{fig:det_flag}
\textit{Edges between $i$ and $j_0$ and $i$ and $k_0$ indicate $j_0 \in \E{i}$ and $ k_0 \in \E{i}$.}
\end{figure}

It is easy to see that the term $\Ujkfunc{i}{c, \bar{s}_{j_0,k_0}}$ depends on $s_{j_0, k_0}=c$, and so does the determinant 
$$\operatorname{det} \left( \left[\Ijkfunc{c,\bar{s}_{j_0,k_0}}\right]_{\E{i}, \E{i}} + \epsilon I \right).$$
Therefore, for each such $i$, 
the contribution to $\widetilde\Lambda_{j_0,k_0}(c, \bar{s}_{j_0,k_0})$ is equal to
$$\operatorname{det} \left( \left[\Ijkfunc{c,\bar{s}_{j_0,k_0}}\right]_{\E{i}, \E{i}} + \epsilon I \right) \times \int_{\cH{i}}\phi\left(\univKKT{i}(c, \bar{s}_{j_0,k_0}, \btil{i}, \zvar{i}); 0_{p-1}, \Omega^{[i]} \right) d\btil{i}.$$
\end{description}
\smallskip

Combining the conclusions from the three possible cases leads us to note that $\widetilde\Lambda_{j_0,k_0}(c, \bar{s}_{j_0,k_0})$ can be replaced by $\Lambda_{j_0,k_0}(c, \bar{s}_{j_0,k_0})$, which completes the proof of the theorem.
\end{proof}

\begin{proof}[Proof of Lemma \ref{lem:monotone}]
The conditional density that we numerically compute to obtain our pivot is equal to
\begin{equation*}
\begin{aligned}
\label{eq:exactlikelihood}
 \frac{ (\operatorname{det}\Ijkfunc{S_{j_0 k_0}, \bar{s}_{j_0,k_0}})^{(n-p-1)/2} \exp \left(  - \theta_{j_0 k_0} S_{j_0 k_0}\right) 
	\widehat\Lambda_{j_0,k_0}(S_{j_0 k_0}, \bar{s}_{j_0,k_0})\cdot\ind_{\SSp}(\Ijkfunc{S_{j_0 k_0}, \bar{s}_{j_0,k_0}})}
{\bigintsss 
(\operatorname{det}\Ijkfunc{t, \bar{s}_{j_0,k_0}})^{(n-p-1)/2} \exp \left(  - \theta_{j_0 k_0} t\right) 
	\widehat\Lambda_{j_0,k_0}(t, \bar{s}_{j_0,k_0})\cdot\ind_{\SSp}(\Ijkfunc{t, \bar{s}_{j_0,k_0}})dt}
\end{aligned}	
\end{equation*}
when evaluated at $c= S_{j_0,k_0}$.
Note that this density is an exponential family density:
\begin{align*}
	p(S_{j_0 k_0};\eta)=\exp \left[\eta T(S_{j_0 k_0})-A(\eta)\right] h(S_{j_0 k_0})
\end{align*}
with the sufficient statistic $T(S_{j_0 k_0}) = S_{j_0 k_0}$ and the natural parameter $\eta =-\theta_{j_0k_0}$.
Furthermore, 
\begin{align*}
    \begin{gathered}
    h(c) = \operatorname{det}\Ijkfunc{c, \bar{s}_{j_0,k_0}})^{(n-p-1)/2} 
	\widehat\Lambda_{j_0,k_0}(c, \bar{s}_{j_0,k_0})\cdot\ind_{\SSp}(\Ijkfunc{c, \bar{s}_{j_0,k_0}},\\
  A(-\theta_{j_0k_0}) = \log \int \exp \left(  - \theta_{j_0 k_0} t\right) h(t)dt.
    \end{gathered}
\end{align*}
 Therefore, it admits a monotonic likelihood ratio, that is, for $\eta_0 = -\theta_{0} < \eta_1 = -\theta_{1}$, the likelihood ratio
$p(S_{j_0 k_0};\eta_1) / p(S_{j_0 k_0};\eta_0)$
is a monotonically increasing function in $S_{j_0 k_0}$. 
This implies that for $ c_1>c_0$,
$$
p(c_1;\eta_1)
p(c_0;\eta_0)
>
p(c_0;\eta_1)
p(c_1;\eta_0).
$$

 Now applying the proof of \cite{lee2016exact}, we integrate over $c_0$ on $(-\infty, c),\ c<c_1$ to obtain
\begin{align*}
p(c_1;\eta_1)\widehat\bF_{\cC; \bar{s}_{j_0,k_0}}(c; \theta_0)  
=& \int_{-\infty}^c p(c_1;\eta_1)p(c_0;\eta_0) d c_0\\
>&
\int_{-\infty}^c p(c_0;\eta_1)p(c_1;\eta_0) d c_0 = p(c_1;\eta_0)\widehat\bF_{\cC; \bar{s}_{j_0,k_0}}(c; \theta_1)  .
\end{align*}
Furthermore, integrating $c_1$ on $(c, \infty)$ gives
$$
\left(1-\widehat\bF_{\cC; \bar{s}_{j_0,k_0}}(c; \theta_1)\right) 
\widehat\bF_{\cC; \bar{s}_{j_0,k_0}}(c; \theta_0) 
>
\left(1-\widehat\bF_{\cC; \bar{s}_{j_0,k_0}}(c; \theta_0) \right) 
\widehat\bF_{\cC; \bar{s}_{j_0,k_0}}(c; \theta_1),
$$
and thus $\widehat\bF_{\cC; \bar{s}_{j_0,k_0}}(c; \theta_0) >\widehat\bF_{\cC; \bar{s}_{j_0,k_0}}(c; \theta_1)$ for $\theta_0>\theta_1$. Hence $\widehat\bF_{\cC; \bar{s}_{j_0,k_0}}(c; \theta)$ is monotonically increasing in $\theta$.
\end{proof}


\section{PROMPT Variables}
\label{sec:AppendixB}
\begin{table}[h]
\scriptsize
\centering
\begin{tabular}{@{}cc@{}}
\toprule
                           & Description                                                                                  \\ \midrule
ActivityCalories           & Calories burned from periods above sedentary level, personal average                         \\
BodyBmi                    & Body Mass Index, from the Body Time Series,                                                  \\
BodyWeight                 & Body weight, from the Body Time Series, personal average                                     \\
Calories                   & Calories, from the Activity Time Series, personal average                                    \\
CaloriesBMR                & Only BMR (Basal Metabolic Rate) calories, from the Activity Time Series, personal average    \\
Distance                   & Distance traveled, from the Activity Time Series, personal average                           \\
HeartRateIntradayCount     & The number of intraday heart rate samples collected during the time period, personal average \\
HeartRateZoneOutOfRangeMax & The out of range max heart rate zone, personal average                                       \\
HeartRateZoneCardioMax     & The Cardio range max heart rate zone, personal average                                       \\
HeartRateZoneCardioMin     & The Cardio range min heart rate zone, personal average                                       \\
ActivityCaloriesSD         & Calories burned from periods above sedentary level, personal sd                              \\
CaloriesSD                 & Calories, from the Activity Time Series, personal average, personal sd                       \\
CaloriesBMR\_SD            & Only BMR (Basal Metabolic Rate) calories, from the Activity Time Series, personal sd         \\
DistanceSD                 & Distance traveled, from the Activity Time Series, personal average, personal sd              \\
HeartRateIntradayCountSD   & The number of intraday heart rate samples collected during the time period, personal sd      \\
ASSIST                     & The Alcohol, Smoking and Substance Involvement Screening Test score                                                                                 \\
GAD                        &  General Anxiety Disorder Survey score \\
ISEL                       & Interpersonal Support Evaluation List score                                                                                   \\
NEO                        & NEO Personality Inventory score                                                                                    \\
PANSI                      & Positive and Negative Suicide Ideation score                                                                                  \\
PCL                        & PTSD Checklist score                                                                                    \\
PHQ                        & Patient Health Questionnaire score                                                                                    \\
PSQI                       & Pittsburgh Sleep Quality Index score                                                                                   \\
RFQ                        & The Reflective Functioning Questionnaire score                                                                                    \\ \bottomrule
\end{tabular}
\caption{Documentation for variables included in the PROMPT analysis}
\label{table:varnames}
\end{table}

\end{appendix}

\end{document}